\documentclass[12pt, draftclsnofoot, onecolumn]{IEEEtran}
\IEEEoverridecommandlockouts
\usepackage{amsfonts}
\usepackage{amssymb}
\usepackage{array, makecell}
\usepackage{mathtools}
\usepackage{graphicx}
\usepackage{subfigure}
\usepackage{enumerate}
\usepackage{amsmath}
\usepackage{color}
\usepackage{amsthm}
\usepackage{algorithm}
\usepackage{algpseudocode}
\usepackage{stfloats}
\usepackage{multirow}
\theoremstyle{definition}
\usepackage{color}
\allowdisplaybreaks[4]

\newcommand{\bp}{\begin{proof} \small }
\newcommand{\ep}{\end{proof} \normalsize}
\newcommand{\epx}{\end{proof} \small}
\newcommand{\bpa}{\begin{proofappx} \footnotesize }
\newcommand{\epa}{\end{proofappx} \small }
\newtheorem{theorem}{Theorem}

\newtheorem{lemma}{Lemma}
\newtheorem{assumption}{Assumption}

\newtheorem{remark}{Remark}
\newtheorem*{theorem*}{Theorem}
\newtheorem*{proposition*}{Proposition}
\newtheorem*{corollary*}{Corollary}
\newtheorem*{lemma*}{Lemma}
\newtheorem*{assumption*}{Assumption}
\newtheorem*{definition*}{Definition}
\newtheorem*{claim*}{Claim}
\newtheorem*{remark*}{Remark}

\newcommand{\be}{\begin{equation}}
\newcommand{\ee}{\end{equation}}
\newcommand{\bs}{\begin{subequations}}
\newcommand{\es}{\end{subequations}}
\newcommand{\bq}{\begin{eqnarray}}
\newcommand{\eq}{\end{eqnarray}}
\newcommand{\bqn}{\begin{eqnarray*}}
\newcommand{\eqn}{\end{eqnarray*}}

\newcommand{\ba}{\left[ \begin{array}}
\newcommand{\ea}{\\ \end{array} \right]}
\newcommand{\ben}{\begin{enumerate}}
\newcommand{\een}{\end{enumerate}}

\def\a{{\boldsymbol{a}}}

\def\real{{\mathchoice%
{\hbox{\rm\setbox1=\hbox{I}\copy1\kern-.45\wd1 R}}
{\hbox{\rm\setbox1=\hbox{I}\copy1\kern-.45\wd1 R}}
{\hbox{\scriptsize\rm\setbox1=\hbox{I}\copy1\kern-.45\wd1 R}}
{\hbox{\scriptsize\rm\setbox1=\hbox{I}\copy1\kern-.45\wd1 R}}}}

\def\Zint{{\mathchoice{\setbox1=\hbox{\sf Z}\copy1\kern-.75\wd1\box1}
{\setbox1=\hbox{\sf Z}\copy1\kern-.75\wd1\box1}
{\setbox1=\hbox{\scriptsize\sf Z}\copy1\kern-.75\wd1\box1}
{\setbox1=\hbox{\scriptsize\sf Z}\copy1\kern-.75\wd1\box1}}}
\newcommand{\complex}{ \hbox{\rm C\kern-0.45em\rule[.07em]{.02em}{.58em}%
\kern 0.43em}}

\makeatletter
\newcommand{\algmargin}{\the\ALG@thistlm}
\makeatother
\newlength{\whilewidth}
\algdef{SE}[parWHILE]{parWhile}{EndparWhile}[1]
{\parbox[t]{\dimexpr\linewidth-\algmargin}{%
		\hangindent\whilewidth\strut\algorithmicwhile\ #1\ \algorithmicdo\strut}}{\algorithmicend\ \algorithmicwhile}%
\algnewcommand{\parState}[1]{\State%
	\parbox[t]{\dimexpr\linewidth-\algmargin}{\strut #1\strut}}

\ifodd 1

\else

\fi

\def\BibTeX{{\rm B\kern-.05em{\sc i\kern-.025em b}\kern-.08em
		T\kern-.1667em\lower.7ex\hbox{E}\kern-.125emX}}
		
\begin{document}	
		
\title{Joint Client Assignment and UAV Route Planning for Indirect-Communication Federated Learning \footnote{An earlier version of this paper was presented at the 57th Annual Conference on Information Sciences and Systems (CISS) \cite{10089783}.}}

\author{
  Jieming Bian, ~Cong Shen, ~Jie Xu \footnote{J. Bian and J. Xu are with the Department of Electrical and Computer Engineering, University of Miami, FL, USA. Email: \{jxb1974, jiexu\}@miami.edu.;  C. Shen is with Charles L. Brown Department of Electrical and Computer Engineering, University of Virginia, VA, USA. Email: cong@virginia.edu.}
}

\maketitle

\begin{abstract}
Federated Learning (FL) is a machine learning approach that enables the creation of shared models for powerful applications while allowing data to remain on devices. This approach provides benefits such as improved data privacy, security, and reduced latency. However, in some systems, direct communication between clients and servers may not be possible, such as remote areas without proper communication infrastructure. To overcome this challenge, a new framework called FedEx (Federated Learning via Model Express Delivery) is proposed. This framework employs mobile transporters, such as UAVs, to establish indirect communication channels between the server and clients. These transporters act as intermediaries and allow for model information exchange. The use of indirect communication presents new challenges for convergence analysis and optimization, as the delay introduced by the transporters' movement creates issues for both global model dissemination and local model collection. To address this, two algorithms, FedEx-Sync and FedEx-Async, are proposed for synchronized and asynchronized learning at the transporter level. Additionally, a bi-level optimization algorithm is proposed to solve the joint client assignment and route planning problem. Experimental validation using two public datasets in a simulated network demonstrates consistent results with the theory, proving the efficacy of FedEx.

\end{abstract}

\section{Introduction}
Federated Learning (FL) is an innovative approach to machine learning that enables a large number of clients to collaborate on the development of a shared model without sharing their private data. The rise of edge devices like smartphones and IoT devices has led to an explosion in data creation, but concerns over privacy have made it challenging to upload this data to the cloud for analysis. FL addresses these concerns by allowing data to remain on the devices where it was created while still enabling the creation of shared models that can be used for powerful machine learning applications. By doing so, FL offers a range of benefits, including greater data privacy, improved data security, and reduced latency. FL is seen as a promising approach for a range of applications, including healthcare, agriculture, transportation, industrial IoT, and mobile applications, because it allows for distributed, privacy-preserving machine learning.

One key aspect that sets FL apart from other distributed learning frameworks is its focus on efficient communication between the clients and the server. Local stochastic gradient descent (SGD) \cite{stich2018local}  and FedAvg \cite{mcmahan2017communication} algorithms allow clients to perform multiple local SGD iterations on their own datasets before uploading the results to the parameter server for aggregation, resulting in more efficient communication compared to earlier distributed learning algorithms such as distributed SGD \cite{balcan2012distributed} , which requires local computation results to be uploaded to the server after every iteration. This emphasis on communication efficiency has inspired the development of many subsequent FL algorithms based on different optimization methods, and makes FL well-suited for use in communication-constrained learning environments where bandwidth is limited or communication patterns are sporadic.

Despite the differences in the adopted optimization methods and the focused settings, most existing FL works have implicitly assumed that clients can directly communicate with the server. Most synchronous FL algorithms often treat the communication frequency between the server and clients as a design parameter, assuming that the communication channel is universally available when needed. Asynchronous FL algorithms  \cite{avdiukhin2021federated, basu2019qsparse} may allow for less regular communication between the server and clients, but still assume direct communication. However, in many real-world systems, clients may not be able to directly communicate with the server at all due to the lack of a proper communication infrastructure. For example, sensing systems (for, e.g., natural resource management, earthquake/tsunami monitoring and forest fire tracking) are often deployed in remote areas with an extremely limited or even no communication infrastructure \cite{fadlullah2021smart}. In these systems, distributed smart sensors collect versatile and rich data over a large area, and spatial and temporal analysis is required to provide high-quality and reliable analysis for actionable intelligence. It is thus imperative to understand whether FL can still work without direct server-client communications and how to optimize the FL algorithms in these settings. 

In this paper, we present a new FL framework called FedEx (Federated Learning via Model Express Delivery) for use in situations where direct communication between the server and clients is not possible. To address this challenge, FedEx employs mobile transporters, such as Unmanned Aerial Vehicles (UAVs), to establish indirect communication channels between the server and clients for the exchange of model information. The mobile transporters act as intermediaries between the server and clients, similar to delivery trucks in a traditional package delivery system. An example of the use of FedEx in smart sensing in remote areas with no communication infrastructure is illustrated in Fig. \ref{sensing-app}.

However, the use of indirect communication in FedEx presents significant new challenges for convergence analysis and optimization. The delay introduced by the time it takes for mobile transporters to move from one location to another creates issues for both the global model dissemination phase and the local model collection phase. It's unclear if FedEx can still converge under this delay, and if so, how fast. Additionally, the transporter scheduling policy can result in either synchronized or asynchronized learning at the transporter level, which further complicates convergence analysis. Finally, Clients encounter different delays based on their location and the transporter routes chosen. This paper considers transporters with various moving speeds and energy constraints, making it nontrivial to evaluate the performance of FedEx concerning transporter route planning  and client assignments under energy constraints. The main contribution of this paper can be summarized as followings:
\begin{itemize}
    \item We propose the FL framework via indirect server-client communications. Two algorithms, coined FedEx-Sync and FedEx-Async, are proposed depending on whether the transporters synchronize their tours among the assigned clients. 
    \item Our proof demonstrates that both FedEx-Sync and FedEx-Async can attain a convergence rate of $O(\frac{1}{\sqrt{NT}})$, where $N$ represents the number of clients and $T$ represents the number of communication slots. This rate matches the convergence rate found in classical FL.
    \item The energy consumption of each transporter is formulated, taking into account the convergence bounds and energy constraints specific to each transporter. Subsequently, a bi-level optimization algorithm is proposed to address the problem of jointly assigning clients and planning routes. 
    \item The experiments results using two public datasets validate the efficacy of FedEx and are consistent with our theory. 
\end{itemize}
The rest of the paper is organized as follows. Section II discusses related work. Section III formulates the problem. Section IV  proposes the FedEx framework and formulates the energy consumption. Section V provides convergence analysis. Section VI addresses the joint client assignment and route planning problem. Section VII provides experiment results. Finally, Section VIII concludes the paper.

\begin{figure}[tb]
	\centering
	\includegraphics[width=0.5\linewidth]{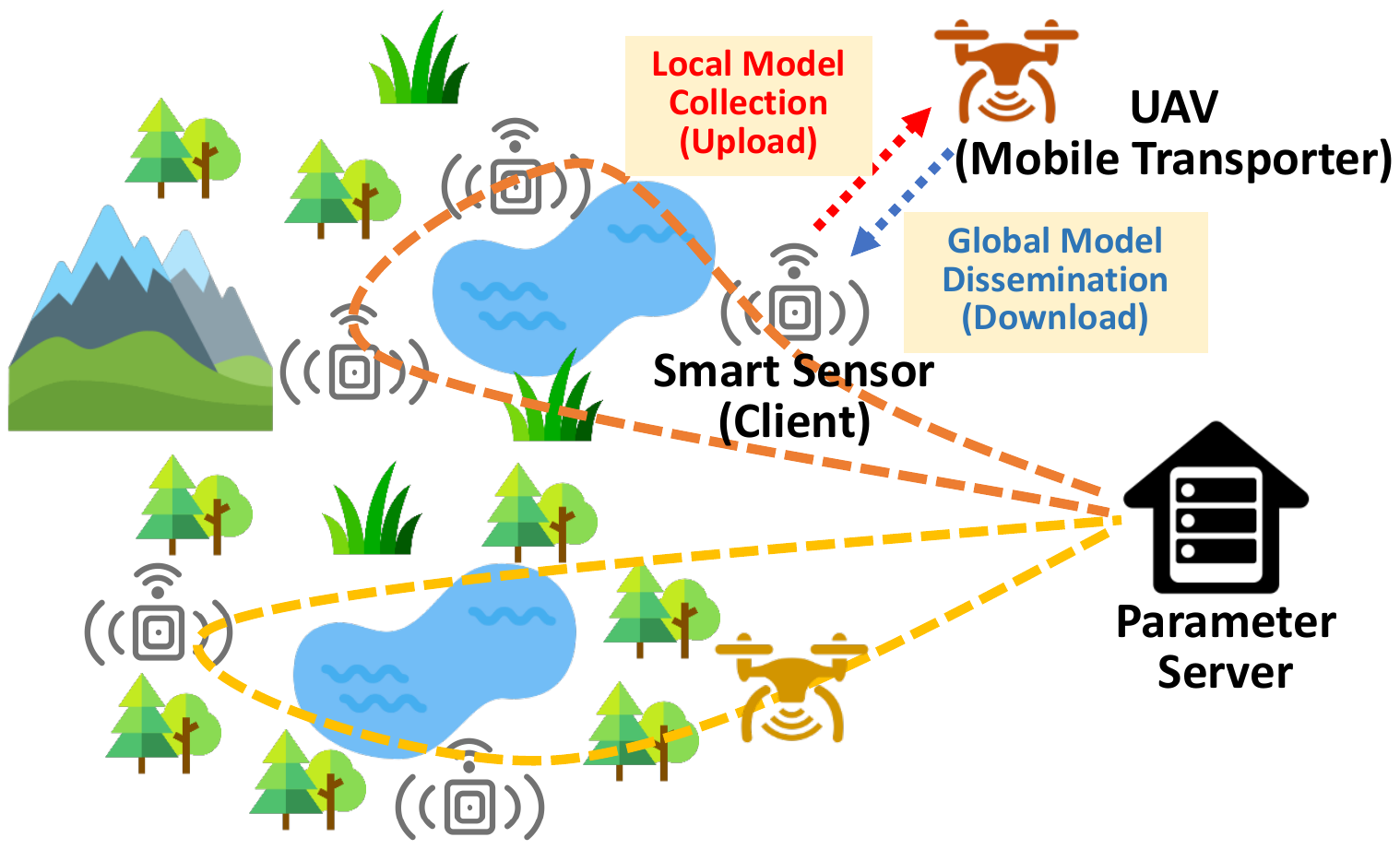}
	\caption{Illustration of FedEx applied to smart sensing in remote areas with no communication infrastructure.} \label{sensing-app}
	\vspace{-15 pt}
\end{figure}%

\section{Related Work}
\textbf{FL via Direct Communications}. Local SGD \cite{stich2018local} and FedAvg \cite{mcmahan2017communication} are among the first FL algorithms, which inspired many follow-up works \cite{li2019convergence, khaled2020tighter, yu2019parallel, wang2019slowmo, liu2020accelerating, karimireddy2020scaffold} based on various optimization techniques in this research area. 
The majority of works on FL study the synchronous setting where clients synchronously upload their local model updates to the server. Asynchronous FL has also been studied in the literature \cite{avdiukhin2021federated, basu2019qsparse}, albeit much smaller. 
In this regard, existing works \cite{chen2020asynchronous, chai2021fedat} usually use asynchronous model aggregation to address the ``straggler'' problem encountered in the synchronous setting. Some others \cite{avdiukhin2021federated, yang2022anarchic} consider asynchronous FL under arbitrary client-server communication patterns.  
However, all these works made an implicit assumption that the model data exchange between the clients and the server is made through a direct communication channel, which may not even exist in scenarios without a communication infrastructure. 

\textbf{Mobility and Relaying in FL}.
Using peer clients as relays to upload local model updates was studied in \cite{yemini2022semi} to combat the unstable wireless channels. However, this work requires not only direct (unstable) server-client channel but also D2D communication channels among static clients. The role of mobility in (asynchronous) FL was first studied in \cite{bian2022mobility}, which uses mobile relay clients to enable more frequent server-client communications. Mobility in \cite{bian2022mobility}, however, is passive since how mobile clients move is not controllable by the algorithm. In our work, FedEx uses dedicated mobile transporters to enable the indirect server-client communications and actively designs the routes of the transporters. Compared to the earlier conference version \cite{10089783}, this work incorporates the transporters' diverse moving speeds and resource energy constraints into the proposed framework for indirect communication FL. This enhancement adds a new dimension to the applicability and practicality of the proposed approach, which can better support real-world scenarios with varying resource availability and mobility constraints.

\textbf{FL with Delayed Averaging}. 
FL with delayed gradient averaging was studied in \cite{zhu2021delayed}. However, the considered setting is very different from ours. Firstly, the work \cite{zhu2021delayed} considers a server-less fully distributed setting where there is a large communication delay among the peer clients. Secondly, the considered delay is only in the (equivalently) global model dissemination phase whereas in our work, delay is also incurred in the local model uploading phase. Furthermore, learning is  synchronous among all clients whereas clients in our setting perform local training asynchronously.

\textbf{UAV Route Planning }. There is a significant body of literature on designing UAV routes for efficient data collection under resource constraints, such as \cite{samir2019uav,hu2020meta,zhao2021multi,zhu2021uav}. However, our work focuses on tailoring the design of mobile transporters with resource constraints to the indirect communication FL problem, which previous works have not addressed.

\section{Problem Formulation}
We consider a Federated Learning (FL) system comprising of one parameter server and $N$ clients. The clients are dispersed across a wide geographical area and lack direct communication capabilities, which means that there is no direct communication channel between any client and the server. For ease of notation, we denote the server as $0$ and the clients as the set $\mathcal{N} = {1, 2, \cdots, N}$. Both the server and clients are stationary and have fixed locations. Based on the location coordinates of the server and clients, a symmetric distance matrix $D \in \mathbb{R}^{(N+1)\times (N+1)}$ can be easily computed, which represents the distance between any two devices. Specifically, $D_{0i}=D_{i0}$ denotes the distance between the server and client $i$, while $D_{ij} = D_{ji}$ denotes the distance between clients $i$ and $j$.

In our scenario, each client $i$ possesses a dataset, and collectively, the clients must train a machine learning model under the coordination of the server by solving the following distributed optimization problem:
\begin{align}
\min_x f(x) = \frac{1}{N}\sum_{i=1}^N f_i(x) = \frac{1}{N}\sum_{i=1}^N \mathbb{E}_{\zeta_i}[F_i(x, \zeta_i)],
\end{align}
Here, $f_i: \mathbb{R}^d \to \mathbb{R}$ is a non-convex loss function specific to client $i$, $F_i$ is the estimated loss function based on a mini-batch data sample $\zeta_i$ drawn from client $i$'s dataset, and $x \in \mathbb{R}^d$ is the model parameter to be learned.

Typically, conventional Federated Learning (FL) frameworks require periodic or non-periodic communication between the clients and the server to train such a machine learning model. However, in our case, direct communication channels between the clients and the server do not exist, rendering such FL frameworks useless. In the next section, we introduce a novel FL framework that enables FL in such extreme communication scenarios (i.e., no direct communications).

\section{FedEx: FL via Model Express Delivery}
\subsection{FedEx Framework}
To overcome the communication challenges arising from the absence of direct communication channels between the clients and the server, we propose a novel approach that involves leveraging mobile transporters, such as UAVs, to establish indirect communication channels. These transporters operate by transporting global and local models between the server and clients, similar to how delivery trucks transport parcels between warehouses and customers. We refer to this new Federated Learning (FL) framework as FedEx, which stands for FL via Model Express Delivery.

Assuming $K$ mobile transporters are available for use in FedEx, we can partition the clients into $K$ non-overlapping subsets and assign each subset to one mobile transporter. Let $\mathcal{R}_k \subset \mathcal{N}$ denote the subset of clients covered by transporter $k$. We have $\mathcal{R}_k\cap\mathcal{R}_{k'} = \emptyset, \forall k \neq k'$ and $\cup_{k=1}^K \mathcal{R}_k = \mathcal{N}$. Additionally, let $R_k = |\mathcal{R}_k|$ denote the number of clients in subset $\mathcal{R}_k$. For our proposed framework, we assume that the mobile transporters fly horizontally at a constant altitude $H$ and we discretize time into slots, indexed by $t=0, 1, 2, \cdots$, where each time slot corresponds to the duration of completing one local training step by a client. We also assume that the moving speed of each transporter is constant, but different among transporters. We denote the moving speed of transporter $k$ as $V_k$, measured by the distance traveled per time slot. Once a transporter reaches either the server or client $i$, it hovers above to establish transmission. Our focus in this paper is primarily on building FL via indirect communication channels. Therefore, we simplify the velocity changes of each transporter and ignore the acceleration/deceleration stage between hovering and the constant speed $V_k$. In contrast to the clients and the server, which are fixed and typically have ample energy resources, mobile transporters are energy-restricted. As a result, in the following sections, we concentrate on analyzing the energy consumption of the $K$ transporters.

\subsection{Transporters Energy Cost}
The energy consumption of the transporters can be broken down into two components. The first component is communication-related energy, which includes the energy used for broadcasting the global model to the clients and the energy used for downloading the global model from the server. The second component is propulsion energy, which is required to keep the transporter (e.g., UAV) aloft and to support its mobility.

\subsubsection{Communication energy}
We assume that all transporters have the same transmission power $p$, and transmission occurs when the transporter hovers above the destination, with the transmission distance being constant at the flight altitude $H$. Following \cite{7888557}, we assume that the communication link between the transporter and the destination is dominated by the Line of Sight (LoS) channel, which follows the free-space path loss model. The channel power at a reference distance of $d_0 = 1$m is denoted by $\beta_0$, and the channel can be expressed as $h = \frac{\beta_0}{H^2}$. Using Shannon's channel capacity equation, the transmission rate $r$ can be expressed as:
\begin{align}
r = B\log_2\left(1 + \frac{h p}{BN_0}\right),
\end{align}
where $B$ is the bandwidth for each client, and $N_0$ is the noise power spectral density.

In FedEx, the transporter facilitates indirect model/gradient exchange between the server and clients. Thus, the transmission size is determined by the model size, rather than the size of the large local dataset. As the sizes of models during each transmission are roughly similar, we denote the transmission size as $S$. Then, the transmission time can be calculated as:
\begin{align}
T_\text{trans} = \frac{S}{r} = \frac{S}{B\log_2\left(1 + \frac{h p}{BN_0}\right)}.
\end{align}


In each round trip of transporter $k$, it needs to broadcast the global model to $R_k$ clients. Therefore, the transmission energy for transporter $k$ is given by:
\begin{align}
E_\text{trans}^k = \frac{pSR_k}{B\log_2\left(1 + \frac{\beta_0 p}{H^2BN_0}\right)}.
\end{align}

\subsubsection{Propulsion energy}
The propulsion energy component consists of energy consumption during steady straight-and-level flight (SLF) between each pair of destinations and energy consumption during hovering above each destination. Let $z_{ij}, \forall i,j\in \mathcal{R}_k\cup{0}$ be a binary variable indicating whether a path from device $i$ to device $j$ is included in the tour. The round trip SLF time for transporter $k$ can be calculated as follows:
\begin{align}
T_\text{SLF}^k = \frac{1}{V_k}\sum_{i=0}^N \sum_{j\neq i, j=0}^N D_{ij}z_{ij}.
\end{align}

To ensure that the tour covers all devices exactly once, we have $2(N+1)$ linear equations, which express that each device has exactly one incoming path and one outgoing path:
\begin{align}
\sum_{i=0, i\neq j}^N z_{ij} = 1, \sum_{j=0, j\neq i}^n z_{ij} = 1, \forall j = 0, \cdots, N.
\end{align}

Following \cite{7888557}, we estimate the SLF energy cost for each transporter $k$ as:
\begin{align}
\label{SLF}
E_\text{SLF}^k = T_\text{SLF}^k\left(c_1V_k^3 + \frac{c_2}{V_k} \right),
\end{align}
where $c_1$ and $c_2$ are constant numbers. In Eq. \ref{SLF}, the cubic of the speed $V_k$ is known as the parasitic power for overcoming the parasitic drag due to the aircraft's skin friction. The term that is inversely proportional to the speed $V_k$ is known as the induced power for overcoming the lift-induced drag, which is the resulting drag force due to wings redirecting air to generate lift for compensating the transporter's weight \cite{9acedc215cb44e448dccb898e60f39af}.

In addition to the SLF energy, there is also energy consumption during hovering for model transmission. We denote the power consumed to hover as $p_\text{hover}$. The total hovering energy consumption of transporter $k$ during one round trip is:
\begin{align}
E_\text{hover}^k = R_k T_\text{trans} p_\text{hover}.
\end{align}

Therefore, the total propulsion energy consumption of transporter $k$ is denoted by $E_\text{prop}^k = E_\text{hover}^k + E_\text{SLF}^k$.

To be more realistic, we consider that the energy budget of each transporter $k$ for each collecting round is limited (e.g., the energy budget $E_\text{budget}^k$ is constrained by the battery of transporter $k$). Without loss of generality, we assume that each transporter changes its battery to reset the energy budget once it returns to the server and ignore the time cost during the battery exchange.

\subsection{Round Trip Time}
The round trip time (RTT) of each transporter $k$ includes hovering transmission time and SLF moving time. It can be expressed as:
\begin{align}
\Delta_k = R_k T_\text{trans} + T_\text{SLF}^k
\end{align}
where $R_k$ is the number of clients assigned to transporter $k$, $T_\text{trans}$ is the transmission time, and $T_\text{SLF}^k$ is the SLF time for transporter $k$ to complete its predetermined tour.

Given a set of clients $\mathcal{R}_k$, the first term of $\Delta_k$ becomes constant. Therefore, finding the shortest tour for the given set $\mathcal{R}_k$ is equivalent to minimizing the second term, which is a well-known problem known as the travelling salesman problem. We will discuss how to optimize the assignment of clients to transporters in Section \ref{opt_section}, after understanding the convergence behavior of FedEx. For now, we treat the set of clients assigned to transporter $k$ and the corresponding tour RTT $\Delta_k$ as already decided.

At runtime, we propose two different versions of FedEx, depending on whether the transporters perform their tours synchronously or asynchronously. We name them FedEx-Sync and FedEx-Async, respectively. Fig. \ref{sync-async} illustrates the key differences between the two versions.

\begin{figure}[tb]
	\centering
	\includegraphics[width=0.5\linewidth]{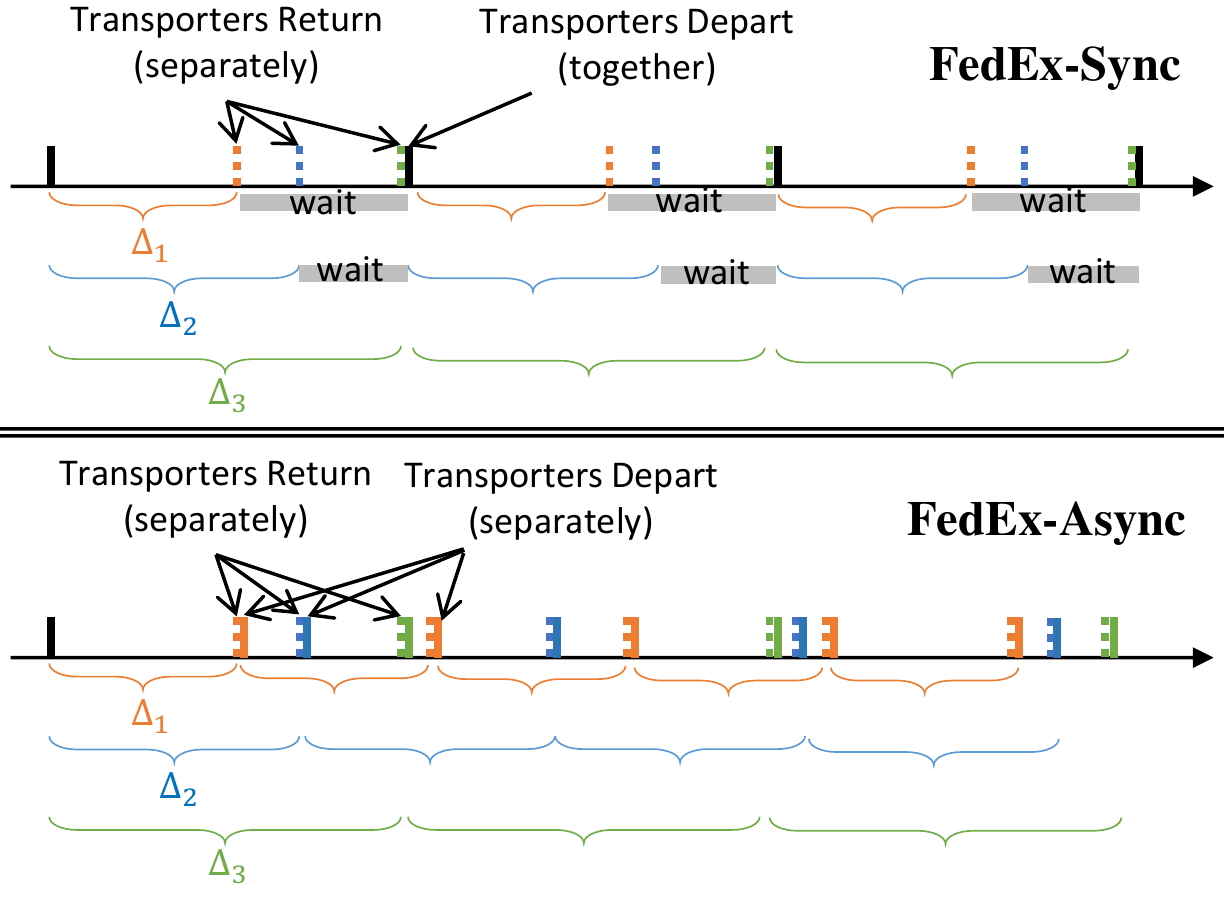}
	\caption{Illustration of the key difference between FedEx-Sync and FedEx-Async.} \label{sync-async}
	\vspace{-15 pt}
\end{figure}%

\subsection{FedEx-Sync}
In the synchronized version of FedEx, namely FedEx-Sync, the transporters depart from the server \textit{at the same time} every time they start a new tour among their assigned clients. Because the transporters have different tour RTTs, the ones with shorter RTT need to wait for the others with longer RTT to come back to start the next tour. Therefore, FedEx-Sync is naturally composed of synchronized learning rounds, with each round having $\Delta \triangleq \max_{k}\Delta_k$ time slots. In each round (denote the first slot of this round as $t_0$), the following events occur.
\begin{itemize}
    \item At the beginning of each round, each transporter downloads the current global model $x^{t_0}$ from the server. The transporters then start a tour among their assigned clients according to the pre-determined client visiting order.
    \item  When a transporter (say transporter $k$) meets a client (say client $i$) at time $t > t_0$, client $i$ downloads the global model, i.e., $x^{t_0}$, that transporter $k$ currently carries. Then the transporter leaves and client $i$ uses $x^t_k = x^{t_0}$ as the initial model to train a new local model using its own local dataset until the next time it meets the transporter. Because the transporter takes $\Delta$ time slots to revisit client $i$, the local training will last $\Delta$ time slots. The local training uses a mini-batch SGD method:
    \begin{align}
        x^{s+1}_i = x^s_i - \eta g^s_i, \forall s = t, \cdots, t + \Delta-1,
    \end{align}
    where $g^s_i = \nabla F_i(x^s_i, \zeta^s_i)$ is the stochastic gradient on a randomly drawn mini-batch $\zeta^s_i$ and $\eta$ is the learning rate. Let $m^t_i \in \mathbb{R}^d$ be the cumulative local updates (CLU) of client $i$ at time $s$ since its last meeting with the transporter, which is updated recursively as follows
    \begin{align}
        m^{s}_i = \sum_{s'=t}^{s-1} \eta g^s_i, \forall s = t, \cdots, t+\Delta.
    \end{align}
    \item When a transporter (say transporter $k$) meets a client (say client $i$) at time $t > t_0$, client $i$ also uploads its current CLU to transporter $k$. Note, however, that this CLU is obtained based on the global model from the \textit{previous} round, i.e., $x^{t_0 - \Delta}$. Transporter $k$ maintains an aggregated CLU $u^t_k$ during the current tour to save storage space and updates it whenever a new client CLU is received according to
    \begin{align}
        u^t_k = u^{t-1}_k + m^t_i.
    \end{align}
    \item When the transporter returns to the server, the aggregated CLU is used to update the global model. In FedEx-Sync, the global model is updated synchronously at the end of each round as follows
    \begin{align}
        x^{t_0+\Delta} = x^{t_0+\Delta - 1} - \frac{1}{N}\sum_{k=1}^K u^{t_0+\Delta - 1}_k.
    \end{align}
\end{itemize}

\subsection{FedEx-Async}
The synchronization in FedEx-Sync is achieved by asking faster transporters to wait for slower transporters. This, however, introduces extra delays for faster transporters. In the case where the slowest transporter takes a tour with a very large RTT, then all the other transporters will have to wait for a long time before starting their next tour. In FedEx-Async, we remove such waiting time by letting the transporter start a new tour immediately after finishing the previous tour. In this way, more clients will be able to perform more frequent global/local model exchanges with the server. FedEx-Async share many similarities with FedEx-Sync and the biggest difference is that each transporter will have \textit{individualized} learning rounds not necessarily synchronized with others. For transporter $k$, its learning round lasts $\Delta_k$ time slots and the following events occur in each round (denote the first slot as $t_0$).
\begin{itemize}
    \item At the beginning of each round, transporter $k$ downloads the current model $x^{t_0}$ from the server. Then it starts its tour among the assigned clients.
    \item When transporter $k$ meets client $i$ at time $t>t_0$, client $i$ downloads $x^{t_0}$ from transporter $k$ and uses $x^t_k = x^{t_0}$ as the initial model to train a new local model until the next time it meets the transporter. Different from FedEx-Syncs, the local training will last $\Delta_k$ time slots, which are different across transporters. 
    \item When transporter $k$ meets client $i$ at time $t > t_0$, client $i$ also uploads its current CLU, which is obtained based on the previous round global model $x^{t_0 - \Delta_k}$, to transporter $k$. Transporter $k$ then updates its aggregated CLU $u^t_k$. 
    \item When the transporter returns to the server, the global model is updated as follows
    \begin{align}
        x^{t_0+\Delta_k} = x^{t_0+\Delta_k - 1} - \frac{1}{N} u^{t_0+\Delta_k - 1}_{k}.
    \end{align}
    Again, this is different from FedEx-Sync since the server does not have to wait for all transporters to update the global model. Note that it is possible that multiple transporters can return to the server in the same time slot (say $t$). In this case, the global model update rule is changed to
    \begin{align}
        x^{t_0+\Delta_k} = x^{t_0+\Delta_k - 1} - \frac{1}{N}\sum_{k' \in \mathcal{S}^{t_0+\Delta_k}} u^{t_0+\Delta_k - 1}_{k'} ,
    \end{align}
    where $\mathcal{S}^{t_0+\Delta_k -1}$ is the set of clients that return to the server at time slot $t_0+\Delta_k -1$. 
\end{itemize}

\section{Convergence Analysis}

\subsection{Aligning Client Training}
Before analyzing the convergence of FedEx, we first describe an equivalent view of FedEx that aligns the local training of clients in the same subset. Consider a learning round of transporter $k$ that carries the global model $x^{t_0}$. Because of the different locations of clients in $\mathcal{R}_k$, the clients receive $x^{t_0}$ and start their new round of local training at different time slots. Once they finish the current round of local training, their CLUs based on $x^{t_0}$ will be uploaded via the transporter to the server at time slot $t_0 + 2\Delta_k$. At that moment, the global model gets an update using these CLUs. 

The unaligned local training of clients, even if covered by the same transporter, would create a major challenge for the convergence analysis of FedEx. Fortunately, there is an equivalent (but imaginary) client training procedure that produces exactly the same global model sequence. Specifically, imagine that clients in $\mathcal{R}_k$ receive the global model $x^{t_0}$ immediately at time slot $t_0$ and perform their local training for  $\Delta_k$ time slots. Then their CLUs are delayed one round to be uploaded to the server. That is, at time slot $t_0 + 2\Delta_k$, the global model gets an update. 
It is clear that the global model update is not affected at all by this change but the local training among clients in the same subset $\mathcal{R}_k$ is now perfectly aligned. Since we are interested in the convergence of the global model, we will consider the equivalent aligned client training procedure in our convergence proof. 

Fig. \ref{equivalent} illustrates this equivalent view, where the local training of the clients is aligned while the global model evolution is unaffected. Essentially, the alignment moves the download delay to the upload phase, but the total delay remains the same. With this change, FedEx-Sync becomes a familiar synchronous FL algorithm but with one round CLU upload delay. In addition to the CLU upload delay, Fed-Async still features asynchronized learning across the client subsets.

\begin{figure}[tb]
	\centering
	\includegraphics[width=0.5\linewidth]{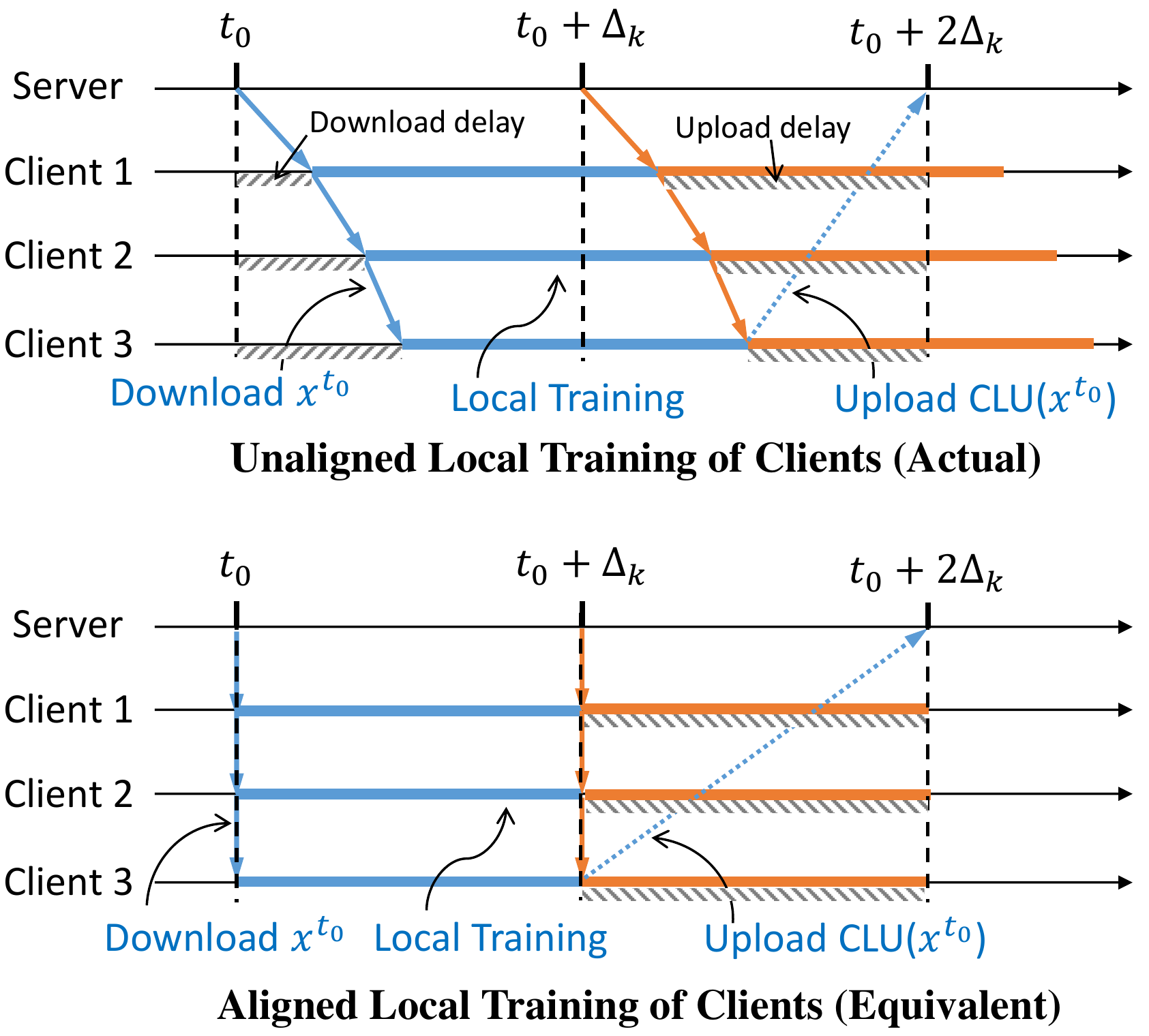}
	\caption{Illustration of the equivalent view of client training. The global model evolution is unaffected by the equivalent view.} \label{equivalent}
	\vspace{-10 pt}
\end{figure}%

\subsection{Assumptions}
Our convergence analysis will utilize the following standard assumptions. 
\begin{assumption}[Lipschitz Smoothness]
There exists a constant $L > 0$ such that $\|\nabla f_i(x) - \nabla f_i(y)\|\leq L \|x - y\|$, $\forall x, y \in \mathbb{R}^d$ and $\forall i = 1, \cdots, N$. 
\end{assumption}

\begin{assumption}[Unbiased Local Gradient Estimate]
The local gradient estimate is unbiased, i.e., $\mathbb{E}_\zeta F_i(x, \zeta) = \nabla f_i(x)$, $\forall x$ and $\forall i = 1, \cdots, N$. 
\end{assumption}

\begin{assumption}[Bounded Gradient]
There exists a constant $G > 0$ such that $\mathbb{E}\|\nabla F_i(x, \zeta)\|^2 \leq G^2$, $\forall x \in \mathbb{R}^d$ and $\forall i = 1, \cdots, N$. 
\end{assumption}

\begin{assumption}[Bounded Variance]
There exists a constant $\sigma > 0$ such that $\mathbb{E}_{\zeta}\|\nabla F_i(x, \zeta) - \nabla f_i(x)\|^2\leq \sigma^2$, $\forall x \in \mathbb{R}^d$ and $\forall i = 1, \cdots, N$. 
\end{assumption}

We do not assume convexity on the loss function $f$ or i.i.d. client data and hence the convergence analysis applies to non-convex and non-i.i.d. problems as well. 

\subsection{Convergence Bound}
The \textit{real sequence} of the global model is the actual global models maintained at the server over time, which can be calculated as follows according to FedEx:
\begin{align}
    x^t = x^0 - \frac{1}{N}\sum_{i=1}^N \sum_{s=0}^{\phi_i(t)}\eta g^s_i, \forall t,
\end{align}
where we define $\phi_i(t)$ as the time slot up to when all corresponding gradients of client $i$ have been received at time $t$. In other words, at time slot $t$, the server has received gradients $g^0_i, \cdots, g^{\phi_i(t)}_i$ from client $i$ (via the transporter). In FedEx-Sync, all clients have the same indirect communication patterns with the server and hence $\phi_i(t) = \phi_j(t), \forall i,j \in \mathcal{N}$. In FedEx-Async, clients belonging to the same transporter have the same indirect communication patterns with the server and hence, $\phi_i(t) = \phi_j(t), \forall i, j\in \mathcal{R}_k, \forall k$. 

\begin{lemma} \label{syn-difference}
For $t > \Delta$, the difference between any client $i$'s local model and the averaged global model is bounded as follows:
\begin{align}
    \mathbb{E}\left[\|x^t - x^t_i\|^2\right] \leq 4\eta^2 G^2 \Delta^2
\end{align}
\end{lemma}

\begin{proof}
The proof is shown in Appendix [\ref{proof_l1}].
\end{proof}

\begin{theorem}\label{Th_1}
For FedEx-Sync, with the bounded initial gradient and learning rate $0<\eta \leq \frac{1}{L} $ , after $T > \Delta$ iterations, we have:

\begin{align}
    &\frac{1}{T}\sum_{t=0}^{T-1} \mathbb{E}\left[\|\nabla f(x^t)\|^2\right]\leq \frac{2}{\eta T}\left(f(x^0) - f* \right) + \frac{\Delta}{T} \|\nabla f(x^0)\|^2+ \frac{T-\Delta}{T}10\eta^2 G^2 L^2 \Delta^2 
    + \frac{T-\Delta}{T}\frac{L\eta\sigma^2}{N} 
\end{align}
\end{theorem}
\begin{proof}
The proof is shown in Appendix [\ref{proof_t1}].
\end{proof}
\begin{remark}
As $T$ is always greatly larger than $\Delta$, we can roughly simplify the bound as 
\begin{align}
    &\frac{1}{T}\sum_{t=0}^{T-1} \mathbb{E}\left[\|\nabla f(x^t)\|^2\right]\leq \frac{2}{\eta T}\left(f(x^0) - f* \right) + 10\eta^2 G^2 L^2 \Delta^2 
    + \frac{L\eta\sigma^2}{N} 
\end{align}
\end{remark}
\begin{remark}
For $T \geq N^3$, by setting the learning rate as $\eta = \frac{\sqrt{N}}{L\sqrt{T}}$, the convergence bound recovers the same $O(\frac{1}{\sqrt{NT}})$ convergence rate of the classical synchronous FL \cite{yu2019parallel}. \end{remark}

Due to the complicated server-transporter-client communication patterns, it is hard to directly analyze the convergence of real sequence in FedEx-Async. Hence, we introduce a \textit{virtual sequence} technique to prove the convergence. 
The \textit{virtual sequence} of the global model is defined in the imaginary case where all client gradients are uploaded to the server immediately after they have been calculated. Similar virtual sequences have been utilized in  \cite{yuan2016convergence, lian2017can, avdiukhin2021federated, nedic2018network}. However, our convergence proof is tailored to the specific problems in our paper and different than all prior works. Specifically, the virtual sequence is defined as
\begin{align}
    v^t = x^0 - \frac{1}{N}\sum_{i=1}^N \sum_{s=0}^{t-1}\eta g^s_i, \forall t.
\end{align}

Clearly, there is a discrepancy between the real sequence and the virtual sequence due to the delayed upload of the client gradients. 
For FedEx-Async, the delay is bounded by
\begin{align}
    (t-1) - \phi_i(t) \leq 2 \Delta_k, \forall i \in \mathcal{R}_k, \forall k.
\end{align}
This can be easily seen from Fig. \ref{equivalent}, 
where the worst case occurs at the first gradient computed by the client. With this delay bound, we can then bound the difference among the virtual global model, the real global model and the real local models on the clients (in the equivalent view) at any time $t$. 
\begin{lemma} \label{lm2}
The difference between the real global model and the virtual global model can be bounded as follows
\begin{align}
    \mathbb{E}\|v^t - x^t\|^2 \leq \frac{4\eta^2 G^2}{N}\sum_{k=1}^K R_{k} \Delta^2_{k}.
     \label{r-v-global-bound}
\end{align} 
The average difference between all clients' local models and the virtual global model is bounded as follows
\begin{align}
    \frac{1}{N}\sum_{i=1}^N \mathbb{E}\|v^t - x^t_i\|^2 \leq \frac{18\eta^2 G^2}{N}
    \sum_{k=1}^K R_{k} \Delta^2_{k}.
\end{align}
\end{lemma}
\begin{proof}
The proof is shown in Appendix [\ref{proof_l2}].
\end{proof}
\begin{theorem}
\label{t2}
For FedEx-Async, by setting the learning rate $0 < \eta \leq 1/L$, we have
\begin{align}
    &\frac{1}{T}\sum_{t=0}^{T-1}\mathbb{E}\|\nabla f(x^t)\|^2  
    \leq \frac{4}{\eta T} (f(x^0) - f^*) 
    + \frac{44 \eta^2 G^2 L^2}{N}\sum_{k=1}^K R_k \Delta^2_k + \frac{2L\eta\sigma^2}{N}.
\end{align}
\end{theorem}

\begin{proof}
The proof is shown in Appendix [\ref{proof_t2}].
\end{proof}

\begin{remark}
For $T \geq N^3$, by setting the learning rate as $\eta = \frac{\sqrt{N}}{L\sqrt{T}}$, the convergence bound recovers the same $O(\frac{1}{\sqrt{NT}})$ convergence rate of the classical synchronous FL \cite{yu2019parallel}. \end{remark}

\begin{remark}
The convergence bound depends on the RTTs of all transporters, i.e., $\Delta_1, \cdots, \Delta_K$, which further depend on how clients are assigned to the transporters. Specifically, for FedEx-Sync, the bound depends on $\max_k \Delta_k$, and for FedEx-Async, the bound depends on $\sum_{k=1}^K R_k \Delta^2_k$. These bounds shed lights on how to assign clients among the transporters to accelerate learning. 
\end{remark}

\section{Client Assignment and Route Planning } \label{opt_section}
In the previous section, we proved that both versions of FedEx converge for \textit{any} client assignment among the transporters. However, the convergence bound depends on \textit{how} clients are assigned to the transporters. In this section, we study the joint client assignment and route planning problem to optimize the convergence bound of FedEx. 

\subsection{Problem Formulation}
We consider a common scenario where the number of clients greatly exceeds the number of transporters, i.e., $N \gg K$, due to limited transporter availability and the potential deployment of a large number of IoT devices. Let $a_i \in {1, \cdots, K}$ represent the assignment variable for client $i$, indicating which transporter it is assigned to. We also gather the assignment variables for all clients in $\a = (a_1, \cdots, a_N)$. Clearly, $\mathcal{R}_k = {i: a_i = k}$.

With the assigned clients $\mathcal{R}_k$ for each transporter $k$, we can devise a route to minimize the RTT. Let $\Delta_k(\mathcal{R}_k)$ be the minimum RTT given a set of clients $\mathcal{R}_k$. Alternatively, $\Delta_k(\mathcal{R}_k)$ can also be expressed as $\Delta_k(\a)$ since $\mathcal{R}_k$ is determined by $\a$. Likewise, we denote $E_\text{prop}^k(\mathcal{R}_k)$ as the total propulsion energy consumption and $E_\text{trans}^k(\mathcal{R}_k)$ as the total transmission energy consumption, given the assigned set $\mathcal{R}_k$ for transporter $k$.

Client assignment aims to solve the following optimization problems with energy consumption constraints:
\begin{align}
    \text{FedEx-Sync}: &\min_\a \max_k\Delta_k(\a),\\
    \textrm{s.t.}  \quad & E_\text{prop}^k(\mathcal{R}_k) + E_\text{trans}^k(\mathcal{R}_k) \leq E_\text{budget}^k; \forall k 
\end{align}
\begin{align}
    \text{FedEx-Async}: &\min_\a \sum_{k} R_k(\a) \Delta^2_k(\a), \\
    \textrm{s.t.}  \quad & E_\text{prop}^k(\mathcal{R}_k) + E_\text{trans}^k(\mathcal{R}_k) \leq E_\text{budget}^k; \forall k
\end{align}

The above problem is a difficult combinatorial optimization problem. Next, we propose a new algorithm to solve this problem. 

\subsection{Bi-level Optimization}
To solve the client assignment and route planning  problem formulated in the last subsection, we develop a bi-level optimization algorithm, called CARP (short for Client Assignment and Route Planning ). 

\subsubsection{Inner-level optimization} The inner-level optimization is to solve the minimum RTT given a set of client $\mathcal{R}_k$, for each transporter $k$, namely computing $\Delta_k(\a)$ for a given $\a$. Remember that the $\Delta_k(\a)$ is contracted by two main parts. The first part is the total time for model transmission during each transporter's tour, which has been determined once the assignment be formulated. Hence, the second part is the main object we need to minimize in the inner-level optimization. With the pre-determined constant transporter speed $V_k$, the second part is a classical traveling salesman problem (TSP) \cite{lin1973effective, flood1956traveling, junger1995traveling} . One can either solve for the exact solution using dynamic programming or obtain an approximation solution using a heuristic algorithm. Considering the high complexity of dynamic programming (i.e., $O(2^n n^2)$ where $n$ is the number of nodes), we use a heuristic algorithm 2-OPT \cite{croes1958method}, which has a time complexity of $O(n^2)$, to compute $\Delta_k(\a)$ for a given $\a$ in our implementation. Specifically, 2-OPT starts with an initial random route and performs the following operations repeatedly. In each iteration, it randomly selects two edges and checks if switching the their connections can shorten the overall route. The process terminates if no improvement can be made. We run 2-OPT multiple times using different initial routes in order to escape from the local optimum.  

\subsubsection{Outer-level optimization} 
To tackle the optimization problem of determining the optimal client assignment, we propose using Gibbs Sampling. We use a unified cost function $C(\a)$, which equals $\max_k \Delta_k(\a)$ for FedEx-Sync and $\sum_{k} R_k(\a) \Delta^2_k(\a)$ for FedEx-Async. Our CARP algorithm visits each client in a pre-defined sequence, generates a probability distribution for the assignment decision of the current client while holding other clients' assignment decisions unchanged, and samples a new assignment decision according to this distribution. This process is repeated for a sufficient number of iterations to ensure that the assignment converges to the optimal solution with high probability.

In each iteration $l$, we pick a client $i$ according to the pre-defined order. We define the available transporters set that could be selected at iteration $l$ by client $i$ as $S_i^l$. A transporter $k$ is in the available set if it satisfies the constraint:

\begin{align}
E_\text{prop}^k(\mathcal{R}_k \cup {i}) + E_\text{trans}^k(\mathcal{R}_k \cup {i}) \leq E_\text{budget}^k.
\end{align}

Then, we evaluate the objective function $C(a_i, \a^l_{-i})$ by varying client $i$'s $S_i^l$ possible choices while holding the other clients' assignment decisions $\a^l_{-i}$ fixed. We then update client $i$'s decision according to the Gibbs Sampling distribution:

\begin{align}
a^{l+1}_i \gets k, \text{with probability } \frac{\exp(-C(k, \a^l_{-i}/q))}{\sum_{k'=1}^K \exp(-C(k', \a^l_{-i}/q))},
\end{align}

where $q > 0$ is a parameter of the algorithm. The action associated with a lower cost is more likely to be chosen. According to the theory of Gibbs Sampling, the algorithm converges to the optimal solution, i.e., $\arg\min_\a C(\a)$ with probability 1, when $q_l$ is a decreasing sequence over the iterations such that $q_l \to 0$.
\textbf{Complexity}: In each iteration, CARP executes at most $K$ TSP calculations. Thus, for a total of $L$ iterations, the time complexity is $O(LK2^N N^2)$ if dynamic programming is used to solve the TSP, or $O(LK N^2)$ if 2-OPT is used to solve the TSP. However, we can reduce the complexity by saving and reusing old TSP results if the same TSP problem re-appears in a future iteration. 

\section{Experiments}

\subsection{Learning Setup}
\textbf{DNN Model and Dataset}. We conduct the FL experiments on two public datasets, i.e., FMNIST \cite{xiao2017fashion} and SVHN \cite{netzer2011reading}. For both datasets, we utilize LeNet \cite{lecun1998gradient} as the backbone DNN model, which consists of two convolutional layers and three fully connected layers. 

\textbf{Dataset Splitting}. For the FMNIST experiments, each client possesses 60 training data samples. For the SVHN experiments, each client possesses 800 training data samples. Three data distributions are simulated. 
\begin{itemize}
    \item IID: The clients' data distributions are i.i.d. This is done by uniformly randomly assigning data to the clients.
    \item Non-IID type 1: We use the Dirichlet method to create non-i.i.d. datasets, which is widely applied in FL research, e.g., \cite{chen2020fedbe}. We use a Dirichelet distribution with parameter 0.3 in the FMNIST experiments and 0.5 in the SVHN experiments. 
    \item Non-IID type 2: We also consider a non-i.i.d. distribution that depends on the location of the clients to capture the more realistic situation where clients close to each other generate similar data. In our simulation, clients in the same block receive the same data distribution. Specifically, for each block, we choose a specific label as the main data type (e.g., label 1 for the first block, label 2 for the second block). Then, this type of data is assigned to clients in this block with a probability of 0.7 for FMNIST (0.5 for SVHN), and the rest of the data types are assigned to the clients with a probability of 0.3 for FMNIST (0.5 for SVHN).
\end{itemize}

\subsection{Energy Setup}
\textbf{Network}. We simulate a network without direct communication, consisting of one parameter server and 40 clients distributed over a 2km x 2km area. The entire area is initially divided into 10 equal-sized blocks (i.e., 400 m $\times$ 1000 m), with 4 clients randomly distributed within each block. See Fig. \ref{fig:clients} for an illustration of the server and client deployment. In this scenario, we consider a time slot, defined as the time required for one training round, to be 1 minute.
\begin{figure*}[h]
    \centering
    \begin{minipage}[t]{0.5\linewidth}\includegraphics[width=1\linewidth]{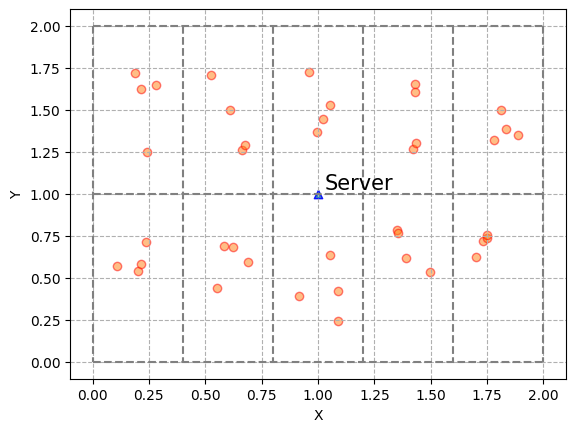}
        \caption{Network Illustration}
	    \label{fig:clients}
    \end{minipage}
\end{figure*}

\textbf{Transporters}. In our main experiments, we simulate $K = 4$ mobile transporters, all moving at the same speed of 36 km/h (10 m/s). We assume that transmissions between transporters and clients occur when the transporters hover 1 km above the clients. The transporters have a transmission bandwidth of 50 Mbps and a transmission power of 20 dBm. The power required for steady straight-and-level flight of the transporters is approximately 30 W, while hovering consumes about 20 W.

\textbf{Routes}. We evaluate different transporter route designs that uses different objective functions to solve the client assignment problem:
\begin{itemize}
    \item Min-Max: $\max_k \Delta_k(\a)$. 
    \item Sum-of-Weighted-Squared (SWS):  $\sum_k R_k(\a)\Delta^2_k(\a)$.
    \item Shortest-Total: $\sum_k \Delta_k(\a)$.
\end{itemize}

\textbf{Energy Consumption and Budget}. We consider energy consumption to include both communication and propulsion energy. For communication energy, although we use a relatively small backbone model in the learning setup, we assume a transmission model size of 100 MB, which is more realistic for modern model sizes (e.g., ResNet34 \cite{he2016deep}, MobileNet \cite{howard2017mobilenets}). Without loss of generality, we assume that each transporter has the same energy budget. In the main experiments, we set the energy budget as 15 KJ for each transporter. We further investigate the effect of lower energy budgets in the following experiments.

\begin{figure*}[h]
    \centering
    \begin{minipage}[t]{0.4\linewidth}\includegraphics[width=1\linewidth]{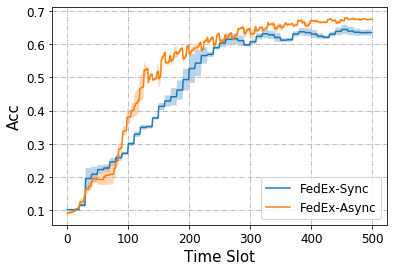}
        \caption{FedEx under i.i.d data on FMNIST}
	    \label{fig:fmnist_iid}
    \end{minipage}
    \begin{minipage}[t]{0.4\linewidth}
        \includegraphics[width=1\linewidth]{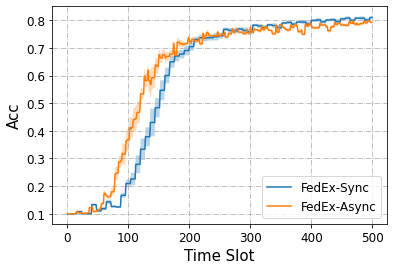}
        \caption{FedEx under i.i.d data on SVHN}
	    \label{fig:svhn_iid}
    \end{minipage}   
\end{figure*}

\subsection{Results}
All results are averaged over four independent experiments. 


\textbf{FedEx-Sync v.s. FedEx-Async under i.i.d. data}. Initially, we evaluate the performance of FedEx-Sync (utilizing Min-Max routes) and FedEx-Async (employing SWS routes) under i.i.d. data conditions. The convergence curves for FMNIST and SVHN are depicted in Fig. \ref{fig:fmnist_iid} and Fig. \ref{fig:svhn_iid}, respectively. It is evident that both algorithms converge under i.i.d. data; however, FedEx-Async surpasses FedEx-Sync in performance. This outcome is reasonable, considering FedEx-Sync's additional time spent waiting for the slowest mobile transporters.

\textbf{FedEx-Sync v.s. FedEx-Async under non-i.i.d. data}. Then we contrast FedEx-Sync and FedEx-Async under two distinct non-i.i.d. data categories. Fig. \ref{fig:fmnist_type1} and Fig. \ref{fig:svhn_type1} indicate that FedEx-Sync's performance is comparable to FedEx-Async's under type-1 non-i.i.d. data. Moreover, Fig. \ref{fig:fmnist_type2} and Fig. \ref{fig:svhn_type2} reveal that FedEx-Sync even surpasses FedEx-Async under type-2 non-i.i.d. data. The decline in FedEx-Async's learning performance is attributed to the client data of a single transporter not representing the overall data distribution in non-i.i.d. settings. Consequently, the transporter's asynchronous learning rounds cause the global model—derived from the partial and \textit{biased} local model update—to deviate from the optimal convergence path, slowing down convergence. Although FedEx-Async saves waiting time compared to FedEx-Sync, additional learning rounds may be required to rectify the deviation, leading to slower convergence for FedEx-Async.

The relative performance of FedEx-Sync and FedEx-Async is heavily influenced by the specific non-i.i.d. data distribution, particularly the degree of bias in a transporter's overall client data. This clarifies why type-2 non-i.i.d. data presents a greater challenge for FedEx-Async than type-1 non-i.i.d. data: under type-1 non-i.i.d. data, even if individual clients' data is highly non-i.i.d., the aggregate client data of a transporter may still resemble the full data distribution when there are enough clients on the transporter's route. However, as clients on the same transporter's route are typically physically close, the overall client data of a transporter remains highly non-i.i.d. in the type-2 case. These experimental findings provide valuable guidance for selecting the appropriate FedEx version in various application scenarios.

\begin{figure*}[h]
    \centering
    \begin{minipage}[t]{0.4\linewidth}
        \includegraphics[width=1\linewidth]{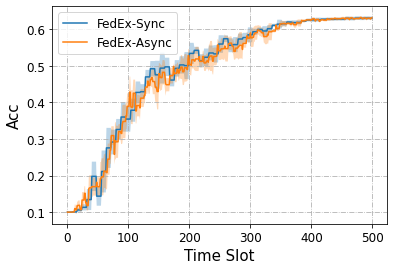}
        \caption{FedEx under type-1 non-i.i.d on FMNIST}
	    \label{fig:fmnist_type1}
    \end{minipage}   
    \begin{minipage}[t]{0.4\linewidth}
        \includegraphics[width=1\linewidth]{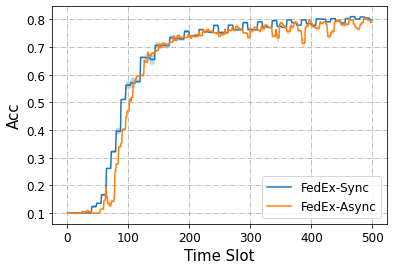}
        \caption{FedEx under type-1 non-i.i.d on SVHN}
	    \label{fig:svhn_type1}
    \end{minipage}
\end{figure*}

\begin{figure*}[h]
    \centering
    \begin{minipage}[t]{0.4\linewidth}
        \includegraphics[width=1\linewidth]{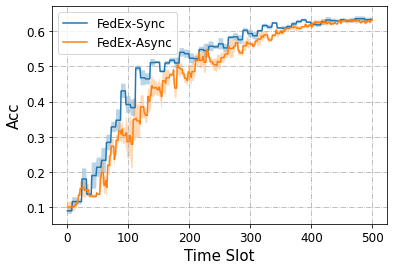}
        \caption{FedEx under type-2 non-i.i.d on FMNIST}
	    \label{fig:fmnist_type2}
    \end{minipage}    
    \begin{minipage}[t]{0.4\linewidth}
        \includegraphics[width=1\linewidth]{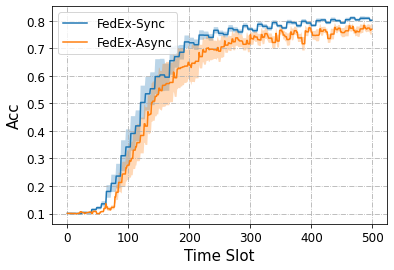}
        \caption{FedEx under type-2 non-i.i.d on SVHN}
	    \label{fig:svhn_type2}
    \end{minipage}   
\end{figure*}

\textbf{Impact of Routes}. In this series of experiments, we explore the effects of various transporter route designs empirically. To minimize the influence of non-i.i.d. data, we focus on the i.i.d. data setting, as FedEx's performance can be highly sensitive to the specific routes chosen and the data distribution under non-i.i.d. conditions. Fig. \ref{fig:asyn-path} compares the convergence of FedEx-Async with different route plannings. The outcome aligns with our theoretical analysis in Theorem 1, where the SWS design delivers the best convergence performance. Fig. \ref{fig:syn-path} presents the results for FedEx-Sync. In this scenario, the longest individual RTT in SWS is the same as that in Min-Max, making SWS an alternative solution to Min-Max. As a result, both Min-Max and SWS routes achieve comparable convergence performance, outpacing Shortest-Total, which is consistent with the theoretical analysis in Theorem 2.

\begin{figure*}
    \centering
    \begin{minipage}[t]{0.4\linewidth}
        \includegraphics[width=1\linewidth]{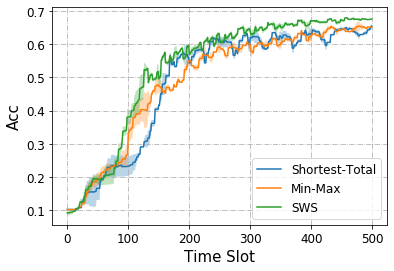}
        \caption{Impact of Routes (FedEx-Async)}
	    \label{fig:asyn-path}
    \end{minipage}
    \begin{minipage}[t]{0.4\linewidth}
        \includegraphics[width=1\linewidth]{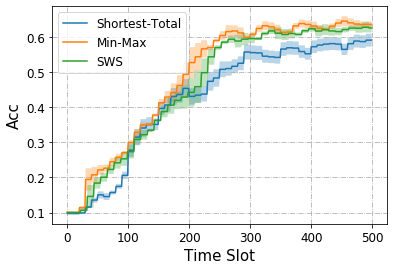}
        \caption{Impact of Routes (FedEx-Sync)}
	    \label{fig:syn-path}
    \end{minipage}    
\end{figure*}

\textbf{Impact of Number of Transporters}. We also evaluate FedEx using different numbers of transporters and present the results in Fig. \ref{fig: different-collectors}. The results show that both FedEx-Sync and FedEx-Async are able to converge for $K = 2$ and $K = 4$. However, increasing the number of transporters results in a faster convergence speed. This can be attributed to two factors. Firstly, with fewer transporters, each transporter has to cover more clients, which increases the time required for disseminating global models and collecting local models. Secondly, with more clients assigned to each transporter, the transmission time is significantly longer, leading to higher RTT for each transporter. Furthermore, FedEx-Async has a smaller advantage over FedEx-Sync when using fewer transporters. This is because the level of asynchronicity is much lower with fewer transporters. In fact, in this experiment, the two transporters have similar RTTs, and hence, FedEx-Async performs similarly to FedEx-Sync.

\begin{figure*}
    \centering
    \begin{minipage}[t]{0.4\linewidth}
        \includegraphics[width=1\linewidth]{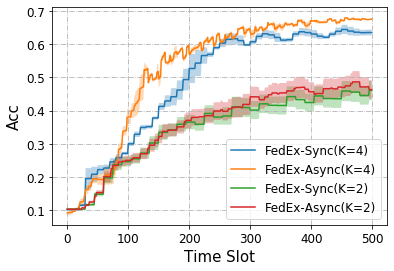}
        \caption{Impact of Number of Transporters}
	    \label{fig: different-collectors}
    \end{minipage}
    \begin{minipage}[t]{0.4\linewidth}
        \includegraphics[width=1\linewidth]{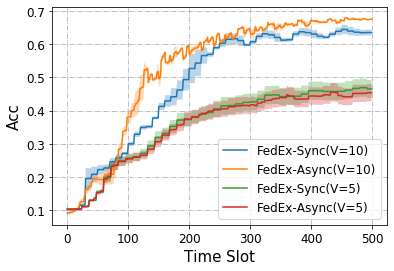}
        \caption{Impact of Moving Speed}
	    \label{fig: different-speeds}
    \end{minipage}    
\end{figure*}

\textbf{Impact of Moving Speed}. We then test FedEx with different transporters' moving speeds. The result can be found in Fig. \ref{fig: different-speeds}. With moving speed $V_k = 5~m/s, \forall k$, both FedEx-Async and FedEx-Sync can still converge. Compared with the scenario, in which each transporter has the moving speed $V_k = 10~m/s$, convergence speeds become slower, as each transporter's RTT increases.

\textbf{Impact of Energy Budget}. We also investigate the effect of different energy budgets on the convergence of FedEx-Sync and FedEx-Async. In the main experiments, we set the energy budget for each transporter as 15 KJ. Here, we reduce the energy budget to 12 KJ and observe the convergence in Fig. \ref{fig: budget_convergence}. Compared with the results with a higher energy budget, the convergence speed of FedEx-Async decreases, while the convergence speed of FedEx-Sync remains similar. This is because with the lower energy budget, the cost function of FedEx-Async fails to reach the optimal point, while the cost function of FedEx-Sync still reaches the optimal point. We plot the energy consumption of each transporter under two different energy budgets of FedEx-Async in Fig. \ref{fig: Energy_budget}. It can be seen that under a low energy budget, the energy consumption of each transporter becomes more similar, indicating that each transporter should take a similar workload. This further causes the cost function to only reach the sub-optimal point.

\begin{figure*}
    \centering
    \begin{minipage}[t]{0.4\linewidth}
        \includegraphics[width=1\linewidth]{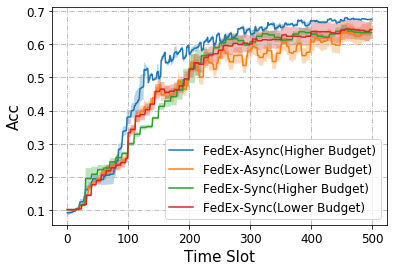}
        \caption{Impact of Energy Budget}
	    \label{fig: budget_convergence}
    \end{minipage}
    \begin{minipage}[t]{0.4\linewidth}
        \includegraphics[width=1\linewidth]{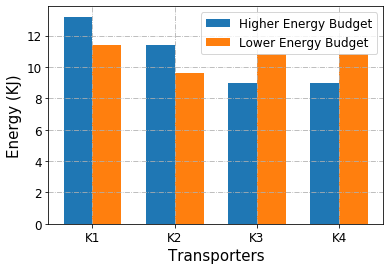}
        \caption{Energy Consumption of each transporters}
	    \label{fig: Energy_budget}
    \end{minipage}    
\end{figure*}

\section{Conclusion}
In this paper, we have presented a novel FL framework named FedEx that utilizes mobile transporters for global model dissemination and local model collection in distributed machine learning scenarios where a communication infrastructure is absent. We have considered various transporters' moving speeds and energy constraints and proposed two new algorithms for efficient client assignment and route planning. We have also conducted a rigorous convergence analysis of these algorithms under non-convex loss functions and non-i.i.d. data distributions. The results provide a systematic approach to solve the joint client assignment and route planning  problem. The performance and convergence of FedEx have been evaluated in numerical experiments on publicly available datasets. Overall, our proposed framework shows promising potential for FL in resource-constrained settings.

\bibliographystyle{IEEEtran}
\bibliography{reference}

\begin{thebibliography}{10}
\providecommand{\url}[1]{#1}
\csname url@samestyle\endcsname
\providecommand{\newblock}{\relax}
\providecommand{\bibinfo}[2]{#2}
\providecommand{\BIBentrySTDinterwordspacing}{\spaceskip=0pt\relax}
\providecommand{\BIBentryALTinterwordstretchfactor}{4}
\providecommand{\BIBentryALTinterwordspacing}{\spaceskip=\fontdimen2\font plus
\BIBentryALTinterwordstretchfactor\fontdimen3\font minus
  \fontdimen4\font\relax}
\providecommand{\BIBforeignlanguage}[2]{{%
\expandafter\ifx\csname l@#1\endcsname\relax
\typeout{** WARNING: IEEEtran.bst: No hyphenation pattern has been}%
\typeout{** loaded for the language `#1'. Using the pattern for}%
\typeout{** the default language instead.}%
\else
\language=\csname l@#1\endcsname
\fi
#2}}
\providecommand{\BIBdecl}{\relax}
\BIBdecl

\bibitem{10089783}
J.~Bian, C.~Shen, and J.~Xu, ``Federated learning via indirect server-client
  communications,'' in \emph{2023 57th Annual Conference on Information
  Sciences and Systems (CISS)}, 2023, pp. 1--5.

\bibitem{stich2018local}
S.~U. Stich, ``Local sgd converges fast and communicates little,'' \emph{arXiv
  preprint arXiv:1805.09767}, 2018.

\bibitem{mcmahan2017communication}
B.~McMahan, E.~Moore, D.~Ramage, S.~Hampson, and B.~A. y~Arcas,
  ``Communication-efficient learning of deep networks from decentralized
  data,'' in \emph{Artificial intelligence and statistics}.\hskip 1em plus
  0.5em minus 0.4em\relax PMLR, 2017, pp. 1273--1282.

\bibitem{balcan2012distributed}
M.~F. Balcan, A.~Blum, S.~Fine, and Y.~Mansour, ``Distributed learning,
  communication complexity and privacy,'' in \emph{Conference on Learning
  Theory}.\hskip 1em plus 0.5em minus 0.4em\relax JMLR Workshop and Conference
  Proceedings, 2012, pp. 26--1.

\bibitem{avdiukhin2021federated}
D.~Avdiukhin and S.~Kasiviswanathan, ``Federated learning under arbitrary
  communication patterns,'' in \emph{International Conference on Machine
  Learning}.\hskip 1em plus 0.5em minus 0.4em\relax PMLR, 2021, pp. 425--435.

\bibitem{basu2019qsparse}
D.~Basu, D.~Data, C.~Karakus, and S.~Diggavi, ``Qsparse-local-sgd: Distributed
  sgd with quantization, sparsification and local computations,''
  \emph{Advances in Neural Information Processing Systems}, vol.~32, 2019.

\bibitem{fadlullah2021smart}
Z.~M. Fadlullah and N.~Kato, ``On smart iot remote sensing over integrated
  terrestrial-aerial-space networks: An asynchronous federated learning
  approach,'' \emph{IEEE Network}, vol.~35, no.~5, pp. 129--135, 2021.

\bibitem{li2019convergence}
X.~Li, K.~Huang, W.~Yang, S.~Wang, and Z.~Zhang, ``On the convergence of fedavg
  on non-iid data,'' \emph{arXiv preprint arXiv:1907.02189}, 2019.

\bibitem{khaled2020tighter}
A.~Khaled, K.~Mishchenko, and P.~Richt{\'a}rik, ``Tighter theory for local sgd
  on identical and heterogeneous data,'' in \emph{International Conference on
  Artificial Intelligence and Statistics}.\hskip 1em plus 0.5em minus
  0.4em\relax PMLR, 2020, pp. 4519--4529.

\bibitem{yu2019parallel}
H.~Yu, S.~Yang, and S.~Zhu, ``Parallel restarted sgd with faster convergence
  and less communication: Demystifying why model averaging works for deep
  learning,'' in \emph{Proceedings of the AAAI Conference on Artificial
  Intelligence}, vol.~33, no.~01, 2019, pp. 5693--5700.

\bibitem{wang2019slowmo}
J.~Wang, V.~Tantia, N.~Ballas, and M.~Rabbat, ``Slowmo: Improving
  communication-efficient distributed sgd with slow momentum,'' \emph{arXiv
  preprint arXiv:1910.00643}, 2019.

\bibitem{liu2020accelerating}
W.~Liu, L.~Chen, Y.~Chen, and W.~Zhang, ``Accelerating federated learning via
  momentum gradient descent,'' \emph{IEEE Transactions on Parallel and
  Distributed Systems}, vol.~31, no.~8, pp. 1754--1766, 2020.

\bibitem{karimireddy2020scaffold}
S.~P. Karimireddy, S.~Kale, M.~Mohri, S.~Reddi, S.~Stich, and A.~T. Suresh,
  ``Scaffold: Stochastic controlled averaging for federated learning,'' in
  \emph{International Conference on Machine Learning}.\hskip 1em plus 0.5em
  minus 0.4em\relax PMLR, 2020, pp. 5132--5143.

\bibitem{chen2020asynchronous}
Y.~Chen, Y.~Ning, M.~Slawski, and H.~Rangwala, ``Asynchronous online federated
  learning for edge devices with non-iid data,'' in \emph{2020 IEEE
  International Conference on Big Data (Big Data)}.\hskip 1em plus 0.5em minus
  0.4em\relax IEEE, 2020, pp. 15--24.

\bibitem{chai2021fedat}
Z.~Chai, Y.~Chen, A.~Anwar, L.~Zhao, Y.~Cheng, and H.~Rangwala, ``Fedat: a
  high-performance and communication-efficient federated learning system with
  asynchronous tiers,'' in \emph{Proceedings of the International Conference
  for High Performance Computing, Networking, Storage and Analysis}, 2021, pp.
  1--16.

\bibitem{yang2022anarchic}
H.~Yang, X.~Zhang, P.~Khanduri, and J.~Liu, ``Anarchic federated learning,'' in
  \emph{International Conference on Machine Learning}.\hskip 1em plus 0.5em
  minus 0.4em\relax PMLR, 2022, pp. 25\,331--25\,363.

\bibitem{yemini2022semi}
M.~Yemini, R.~Saha, E.~Ozfatura, D.~G{\"u}nd{\"u}z, and A.~J. Goldsmith,
  ``Semi-decentralized federated learning with collaborative relaying,''
  \emph{arXiv preprint arXiv:2205.10998}, 2022.

\bibitem{bian2022mobility}
J.~Bian and J.~Xu, ``Mobility improves the convergence of asynchronous
  federated learning,'' \emph{arXiv preprint arXiv:2206.04742}, 2022.

\bibitem{zhu2021delayed}
L.~Zhu, H.~Lin, Y.~Lu, Y.~Lin, and S.~Han, ``Delayed gradient averaging:
  Tolerate the communication latency for federated learning,'' \emph{Advances
  in Neural Information Processing Systems}, vol.~34, pp. 29\,995--30\,007,
  2021.

\bibitem{samir2019uav}
M.~Samir, S.~Sharafeddine, C.~M. Assi, T.~M. Nguyen, and A.~Ghrayeb, ``Uav
  trajectory planning for data collection from time-constrained iot devices,''
  \emph{IEEE Transactions on Wireless Communications}, vol.~19, no.~1, pp.
  34--46, 2019.

\bibitem{hu2020meta}
Y.~Hu, M.~Chen, W.~Saad, H.~V. Poor, and S.~Cui, ``Meta-reinforcement learning
  for trajectory design in wireless uav networks,'' in \emph{GLOBECOM 2020-2020
  IEEE Global Communications Conference}.\hskip 1em plus 0.5em minus
  0.4em\relax IEEE, 2020, pp. 1--6.

\bibitem{zhao2021multi}
C.~Zhao, J.~Liu, M.~Sheng, W.~Teng, Y.~Zheng, and J.~Li, ``Multi-uav trajectory
  planning for energy-efficient content coverage: A decentralized
  learning-based approach,'' \emph{IEEE Journal on Selected Areas in
  Communications}, vol.~39, no.~10, pp. 3193--3207, 2021.

\bibitem{zhu2021uav}
B.~Zhu, E.~Bedeer, H.~H. Nguyen, R.~Barton, and J.~Henry, ``Uav trajectory
  planning in wireless sensor networks for energy consumption minimization by
  deep reinforcement learning,'' \emph{IEEE Transactions on Vehicular
  Technology}, vol.~70, no.~9, pp. 9540--9554, 2021.

\bibitem{7888557}
Y.~Zeng and R.~Zhang, ``Energy-efficient uav communication with trajectory
  optimization,'' \emph{IEEE Transactions on Wireless Communications}, vol.~16,
  no.~6, pp. 3747--3760, 2017.

\bibitem{9acedc215cb44e448dccb898e60f39af}
A.~Filippone, \emph{\BIBforeignlanguage{English}{Flight Performance of Fixed-
  and Rotary-Wing Aircraft}}, 1st~ed.\hskip 1em plus 0.5em minus 0.4em\relax
  Netherlands: Elsevier BV, May 2006.

\bibitem{yuan2016convergence}
K.~Yuan, Q.~Ling, and W.~Yin, ``On the convergence of decentralized gradient
  descent,'' \emph{SIAM Journal on Optimization}, vol.~26, no.~3, pp.
  1835--1854, 2016.

\bibitem{lian2017can}
X.~Lian, C.~Zhang, H.~Zhang, C.-J. Hsieh, W.~Zhang, and J.~Liu, ``Can
  decentralized algorithms outperform centralized algorithms? a case study for
  decentralized parallel stochastic gradient descent,'' \emph{Advances in
  Neural Information Processing Systems}, vol.~30, 2017.

\bibitem{nedic2018network}
A.~Nedi{\'c}, A.~Olshevsky, and M.~G. Rabbat, ``Network topology and
  communication-computation tradeoffs in decentralized optimization,''
  \emph{Proceedings of the IEEE}, vol. 106, no.~5, pp. 953--976, 2018.

\bibitem{lin1973effective}
S.~Lin and B.~W. Kernighan, ``An effective heuristic algorithm for the
  traveling-salesman problem,'' \emph{Operations research}, vol.~21, no.~2, pp.
  498--516, 1973.

\bibitem{flood1956traveling}
M.~M. Flood, ``The traveling-salesman problem,'' \emph{Operations research},
  vol.~4, no.~1, pp. 61--75, 1956.

\bibitem{junger1995traveling}
M.~J{\"u}nger, G.~Reinelt, and G.~Rinaldi, ``The traveling salesman problem,''
  \emph{Handbooks in operations research and management science}, vol.~7, pp.
  225--330, 1995.

\bibitem{croes1958method}
G.~A. Croes, ``A method for solving traveling-salesman problems,''
  \emph{Operations research}, vol.~6, no.~6, pp. 791--812, 1958.

\bibitem{xiao2017fashion}
H.~Xiao, K.~Rasul, and R.~Vollgraf, ``Fashion-mnist: a novel image dataset for
  benchmarking machine learning algorithms,'' \emph{arXiv preprint
  arXiv:1708.07747}, 2017.

\bibitem{netzer2011reading}
Y.~Netzer, T.~Wang, A.~Coates, A.~Bissacco, B.~Wu, and A.~Y. Ng, ``Reading
  digits in natural images with unsupervised feature learning,'' 2011.

\bibitem{lecun1998gradient}
Y.~LeCun, L.~Bottou, Y.~Bengio, and P.~Haffner, ``Gradient-based learning
  applied to document recognition,'' \emph{Proceedings of the IEEE}, vol.~86,
  no.~11, pp. 2278--2324, 1998.

\bibitem{chen2020fedbe}
H.-Y. Chen and W.-L. Chao, ``Fedbe: Making bayesian model ensemble applicable
  to federated learning,'' \emph{arXiv preprint arXiv:2009.01974}, 2020.

\bibitem{he2016deep}
K.~He, X.~Zhang, S.~Ren, and J.~Sun, ``Deep residual learning for image
  recognition,'' in \emph{Proceedings of the IEEE conference on computer vision
  and pattern recognition}, 2016, pp. 770--778.

\bibitem{howard2017mobilenets}
A.~G. Howard, M.~Zhu, B.~Chen, D.~Kalenichenko, W.~Wang, T.~Weyand,
  M.~Andreetto, and H.~Adam, ``Mobilenets: Efficient convolutional neural
  networks for mobile vision applications,'' \emph{arXiv preprint
  arXiv:1704.04861}, 2017.

\end{thebibliography}

\appendices
\section{Proof of Lemma \ref{syn-difference}}
\label{proof_l1}
Consider $t'$ be the last time when client $i$ receives the global model, the we have
\begin{align}
    &\mathbb{E}\left[\|x^t - x^t_i\|^2\right]
    =  \mathbb{E}\left[\|x^t - x^{t'} + x^{t'} - x^t_i\|^2\right]\nonumber\\
    \leq & 2 \mathbb{E}\left[\|x^t - x^{t'}\|^2\right] + 2  \mathbb{E}\left[\|x^{t}_i - x^{t'}_i\|^2\right]\nonumber\\
    \leq & 2\mathbb{E}\left[\left\|\frac{1}{N}\sum_{j=1}^N \sum_{s= t'-1-\Delta}^{t-1-\Delta} \eta g^s_j\right\|^2\right] 
     + 2\mathbb{E}\left[\left\| \sum_{s= t'-1}^{t-1} \eta g^s_i\right\|^2\right]\nonumber\\
    \leq & \frac{2\eta^2}{N}(t -t')\sum_{j=1}^N \sum_{s= t'-1-\Delta}^{t-1-\Delta} \mathbb{E}\left[\|g_j^s\|^2\right]
    + 2(t-t')\eta^2 \sum_{s= t'-1}^{t-1}\mathbb{E}\left[\|g_i^s\|^2\right]\nonumber\\
    \leq & \frac{2\eta^2}{N}(t -t')\sum_{j=1}^N \sum_{s= t'-1-\Delta}^{t-1-\Delta} G^2
     + 2(t-t')\eta^2 \sum_{s= t'-1}^{t-1} G^2\nonumber\\
    \leq & 4 \eta^2G^2\Delta^2
\end{align}

\section{Proof of Theorem \ref{Th_1}}
\label{proof_t1}
Although the averaged global model $x^t$ is not observable at $t \neq a\Delta, a = 0,1,2, \dots$. Here we can still utilize $x^t$ for analysis. 
For $t > \Delta$, by the smoothness of $f(x)$, we have

\begin{align}\label{sf_bound}
    \mathbb{E}[\|f(x^{t+1})\|] &\leq \mathbb{E}[\|f(x^t)\|] + \mathbb{E}[\langle \nabla f(x^t), x^{t+1} - x^t\rangle] + \frac{L}{2} \mathbb{E}[\|x^{t+1} - x^t\|^2]
\end{align}
The second term on the right-hand side of \eqref{sf_bound} can be expressed as follows,
\begin{align}
    &\mathbb{E}[\langle \nabla f(x^t), x^{t+1} - x^t\rangle]\nonumber \\
    =& -\eta \mathbb{E}[\langle \nabla f(x^t), \frac{1}{N}\sum_{i=1}^N g^{t-\Delta}_i\rangle]\nonumber \\ 
    =& -\eta \mathbb{E}[\langle \nabla f(x^t), \frac{1}{N}\sum_{i=1}^N \nabla f_i(x^{t-\Delta}_i)\rangle]\nonumber\\
    =& -\frac{\eta}{2}\mathbb{E}\left[\|\nabla f(x^t)\|^2 + \|\frac{1}{N}\sum_{i=1}^N \nabla f_i(x^{t-\Delta}_i)\|^2\right] + \frac{\eta}{2}\mathbb{E}
    \|\nabla f(x^t) - \frac{1}{N}\sum_{i=1}^N \nabla f_i(x^{t-\Delta}_i)\|^2\nonumber\\
    =&-\frac{\eta}{2}\mathbb{E}\left[\|\nabla f(x^t)\|^2\right] -\frac{\eta}{2}\mathbb{E}\left[ \|\frac{1}{N}\sum_{i=1}^N \nabla f_i(x^{t-\Delta}_i)\|^2\right] +\frac{\eta}{2}\mathbb{E}\left[ \|\nabla f(x^t) - \frac{1}{N}\sum_{i=1}^N \nabla f_i(x^{t-\Delta}_i)\|^2\right]
\end{align}

The third term on the right-hand side of \eqref{sf_bound} can be bounded as follows,
\begin{align}
    &\frac{L}{2} \mathbb{E}[\|x^{t+1} - x^t\|^2] = \frac{L\eta^2}{2}\mathbb{E}\left[\|\frac{1}{N}\sum_{i=1}^N g^{t-\Delta}_i\|^2\right]\nonumber\\
    = & \frac{L\eta^2}{2} \mathbb{E}\left[\|\frac{1}{N}\sum_{i=1}^N (g^{t-\Delta}_i - \nabla f_i(x^{t-\Delta}_i))\|^2 \right] +  \frac{L\eta^2}{2}\mathbb{E}\left[\|\frac{1}{N}\sum_{i=1}^N  \nabla f_i(x^{t-\Delta}_i)\|^2 \right]\nonumber\\
    =&\frac{L\eta^2}{2N^2}\sum_{i=1}^N \mathbb{E}\left[\|g^{t-\Delta}_i - \nabla f_i(x^{t-\Delta}_i)\|^2\right]  + \frac{L\eta^2}{2}\mathbb{E}\left[\|\frac{1}{N}\sum_{i=1}^N  \nabla f_i(x^{t-\Delta}_i)\|^2 \right]\nonumber\\
    \leq& \frac{L\eta^2\sigma^2}{2N} + \frac{L\eta^2}{2}\mathbb{E}\left[\|\frac{1}{N}\sum_{i=1}^N  \nabla f_i(x^{t-\Delta}_i)\|^2 \right]
\end{align}

Substituting these into \eqref{sf_bound} yields
\begin{align}
    &\mathbb{E}[\|f(x^{t+1}\|] \nonumber\\
    \leq& \mathbb{E}[\|f(x^t)\|]-\frac{\eta}{2}\mathbb{E}\left[\|\nabla f(x^t)\|^2\right] -\frac{\eta - L\eta^2}{2}\mathbb{E}\left[ \|\frac{1}{N}\sum_{i=1}^N \nabla f_i(x^{t-\Delta}_i)\|^2\right] \nonumber\\
    &+\frac{\eta}{2}\mathbb{E}\left[ \|\nabla f(x^t) - \frac{1}{N}\sum_{i=1}^N \nabla f_i(x^{t-\Delta}_i)\|^2\right] +\frac{L\eta^2\sigma^2}{2N}\nonumber\\
    \leq& \mathbb{E}[\|f(x^t)\|]-\frac{\eta}{2}\mathbb{E}\left[\|\nabla f(x^t)\|^2\right] + \frac{\eta}{2}\mathbb{E}\left[ \|\nabla f(x^t) - \frac{1}{N}\sum_{i=1}^N \nabla f_i(x^{t-\Delta}_i)\|^2\right] +\frac{L\eta^2\sigma^2}{2N}\nonumber\\
    \leq& \mathbb{E}[\|f(x^t)\|]-\frac{\eta}{2}\mathbb{E}\left[\|\nabla f(x^t)\|^2\right] +\frac{\eta}{2N}\sum_{i=1}^N\mathbb{E}\left[\|\nabla f_i(x^t) - \nabla f_i(x^{t-\Delta}_i)\|^2\right] +\frac{L\eta^2\sigma^2}{2N}\nonumber\\
    \leq& \mathbb{E}[\|f(x^t)\|]-\frac{\eta}{2}\mathbb{E}\left[\|\nabla f(x^t)\|^2\right] +\frac{\eta L^2}{2N}\sum_{i=1}^N\mathbb{E}\left[\|x^t - x^{t-\Delta}_i\|^2\right] +\frac{L\eta^2\sigma^2}{2N} \nonumber\\
   \leq& \mathbb{E}[\|f(x^t)\|]-\frac{\eta}{2}\mathbb{E}\left[\|\nabla f(x^t)\|^2\right] +\frac{\eta L^2}{2N}\sum_{i=1}^N\mathbb{E}\left[\|x^t - x^t_i + x^t_i - x^{t-\Delta}_i\|^2\right] +\frac{L\eta^2\sigma^2}{2N} \nonumber\\
    \leq& \mathbb{E}[\|f(x^t)\|]-\frac{\eta}{2}\mathbb{E}\left[\|\nabla f(x^t)\|^2\right] +\frac{L\eta^2\sigma^2}{2N} +\frac{\eta L^2}{2N}\sum_{i=1}^N(2\mathbb{E}\left[\|x^t - x^t_i\|^2\right] + 2 \mathbb{E}\left[\|x^t_i - x^{t-\Delta}_i\|^2\right])\nonumber \\
    \leq& \mathbb{E}[\|f(x^t)\|]-\frac{\eta}{2}\mathbb{E}\left[\|\nabla f(x^t)\|^2\right] +\frac{L\eta^2\sigma^2}{2N} + \frac{\eta L^2}{2N}\sum_{i=1}^N (8\eta^2 G^2 \Delta^2 + 2\eta^2 G^2 \Delta^2)\nonumber\\
     \leq& \mathbb{E}[\|f(x^t)\|]-\frac{\eta}{2}\mathbb{E}\left[\|\nabla f(x^t)\|^2\right] + 5\eta^3 G^2 L^2 \Delta^2 +\frac{L\eta^2\sigma^2}{2N}
\end{align}

Dividing both sides by $\frac{\eta}{2}$ and rearranging the terms, we have
\begin{align}
    \mathbb{E}\left[\|\nabla f(x^t)\|^2\right] &\leq \frac{2}{\eta}\left(\mathbb{E}[\|f(x^t)\|] - \mathbb{E}[\|f(x^{t+1}\|] \right)  + 10\eta^2 G^2 L^2 \Delta^2 + \frac{L\eta\sigma^2}{N}
\end{align}
Note that $x_t$ is unchanged during $t \in \left[0, \Delta\right]$. Taking the sum over $t = 0, 1, ..., T-1$ and dividing both sides by $T$, we have
\begin{align}
    &\frac{1}{T}\sum_{t=0}^{T-1} \mathbb{E}\left[\|\nabla f(x^t)\|^2\right]
    \leq  \frac{2}{\eta T}\left(\mathbb{E}[\|f(x^0)\|] - \mathbb{E}[\|f(x^T\|] \right) \nonumber\\
    &+ \frac{T-\Delta}{T}10\eta^2 G^2 L^2 \Delta^2 + \frac{T-\Delta}{T}\frac{L\eta\sigma^2}{N} + \frac{\Delta}{T}\mathbb{E}\left[\|\nabla f(x^0)\|^2\right]\nonumber\\
     \leq&  \frac{2}{\eta T}\left(f(x^0) - f* \right) 
    + \frac{T-\Delta}{T}10\eta^2 G^2 L^2 \Delta^2 + \frac{T-\Delta}{T}\frac{L\eta\sigma^2}{N} + \frac{\Delta}{T} \|\nabla f(x^0)\|^2
\end{align}
where $f^* = \min_x f(x)$. 
This completes the proof.

\section{Proof of Lemma \ref{lm2}}
\label{proof_l2}
At any time $t$, using the definition of the virtual and real sequences, we have
\begin{align}
    \mathbb{E}\|v^t - x^t\|^2 
    =&\mathbb{E}\|\frac{1}{N}\sum_{i=1}^N \sum_{s= \phi_i(t)}^{t-1}\eta g^s_i\|^2 
    \leq  \frac{\eta^2}{N}\sum_{i=1}^N \mathbb{E}\|\sum_{s=\phi_i(t)}^{t-1} g^s_i\|^2 \nonumber\\
    \leq & \frac{\eta^2}{N}\sum_{i=1}^N((t-1) - \phi_i(t))^2 G^2 \leq \frac{4\eta^2 G^2}{N}\sum_{k=1}^K R_{k} \Delta^2_{k}.
\end{align}

At any time $t$, denote $\tau_k$ as the global model version that clients in $\mathcal{R}_k$ use to train the local models. Consider client $i \in \mathcal{R}_k$, we have
\begin{align}
    &\mathbb{E}\|v^t - x^t_i\|^2 =\mathbb{E}\|v^t - v^{\tau_k} + v^{\tau_k} - x^{\tau_k} + x^{\tau_k} - x^t_i\|^2 \nonumber\\
    \leq & 3\mathbb{E}\left[\|v^t - v^{\tau_k}\|^2 + \|v^{\tau_k} - x^{\tau_k}\|^2 + \|x^{\tau_k} - x^t_i\|^2 \right].
\end{align}
The first term on the right-hand side can be bounded by
\begin{align}
    &\mathbb{E}\|v^t - v^{\tau_k}\|^2 =\mathbb{E}\|\frac{1}{N}\sum_{j=1}^N\sum_{s = \tau_k}^{t-1}\eta g^s_j\|^2 
    \leq  \frac{\eta^2}{N}\sum_{j=1}^N \mathbb{E}\|\sum_{s=\tau_k}^{t-1} g^s_j\|^2
    \leq  \eta^2 G^2 \Delta^2_k.
\end{align}
The second term can be bounded according to \eqref{r-v-global-bound}. The third term can be bounded by
\begin{align}
    &\mathbb{E}\|x^{\tau_k} - x^t_i\|^2 = \mathbb\|\sum_{s=\tau_k}^{t-1} \eta g^s_i\|^2 \leq \eta^2 G^2 \Delta^2_k.
\end{align}
Summing up all three bounds, we have
\begin{align}
    \mathbb{E}\|v^t - x^t_i\|^2 \leq 6\eta^2 G^2 \Delta^2_k + \frac{12 \eta^2 G^2}{N}\sum_{k'=1}^K R_{k'} \Delta^2_{k'}.
\end{align}
Therefore, the average difference is bounded by
\begin{align}
    \frac{1}{N}\sum_{i=1}^N \mathbb{E}\|v^t - x^t_i\|^2 
    \leq& \frac{6\eta^2 G^2}{N} 
    \sum_{k'=1}^K R_{k'} \Delta^2_{k'} + \frac{12\eta^2 G^2}{N} \sum_{k'=1}^K R_{k'} \Delta^2_{k'} = \frac{18\eta^2 G^2}{N}
    \sum_{k'=1}^K R_{k'} \Delta^2_{k'}.
\end{align}

\section{Proof of Theorem \ref{t2}}
\label{proof_t2}
We first analyze the convergence of the virtual global model sequence. By the smoothness of $f(x)$ in Assumption 1, we have
\begin{align}
    \mathbb{E}\|f(v^{t+1)}\| 
    \leq \mathbb{E}\|f(v^t)\| + \mathbb{E}\langle \nabla f(v^t), v^{t+1} - v^t\rangle 
    + \frac{L}{2}\mathbb{E}\|v^{t+1} - v^t\|^2.
    \label{one-step}
\end{align}
The second term on the right-hand side of \eqref{one-step} can be expressed as follows.
\begin{align}
    &\mathbb{E}\langle \nabla f(v^t), v^{t+1} - v^t\rangle \nonumber\\
    =&-\eta \mathbb{E}\langle \nabla f(v^t), \frac{1}{N}\sum_{i=1}^N \nabla f_i(x^t_i)\rangle \nonumber\\
    =& -\frac{\eta}{2}\mathbb{E}\left[\|f(v^t)\|^2 + \|\frac{1}{N}\sum_{i=1}^N \nabla f_i(x^t_i)\|^2\right] +\frac{\eta}{2}\mathbb{E}\|f(v^t) - \frac{1}{N}\sum_{i=1}^N \nabla f_i(x^t_i)\|^2.
\end{align}
The third term on the right-hand side of \eqref{one-step} can be bounded as follows
\begin{align}
    &\frac{L}{2}\mathbb{E}\|v^{t+1} - v^t\|^2 = \frac{L\eta^2}{2}\mathbb{E}\|\frac{1}{N}\sum_{i=1}^N g^t_i\|^2 \nonumber\\
    = & \frac{L\eta^2}{2}\mathbb{E}\|\frac{1}{N}\sum_{i=1}^N (g^t_i - \nabla f_i(x^t_i)) + \frac{1}{N}\sum_{i=1}^N \nabla f_i(x^t_i)\|^2 \nonumber\\
    \leq & \frac{L\eta^2\sigma^2}{2N} + \frac{L\eta^2}{2}\mathbb{E}\|\frac{1}{N}\sum_{i=1}^N \nabla f_i(x^t_i)\|^2.
\end{align}
Substituting these into \eqref{one-step} yields
\begin{align}
    \mathbb{E}\|f(v^{t+1}\| 
    \leq &\mathbb{E}\|f(v^t)\| - \frac{\eta}{2}\mathbb{E}\|\nabla f(v^t)\|^2 \nonumber - \frac{\eta - L\eta^2}{2}\mathbb{E}\|\frac{1}{N}\sum_{i=1}^N \nabla f_i(x^t_i)\|^2 \nonumber\\
    & + \frac{\eta}{2}\mathbb{E}\|f(v^t) - \frac{1}{N}\sum_{i=1}^N \nabla f_i(x^t_i)\|^2 + \frac{L\eta^2\sigma^2}{2N} \nonumber \\
    \leq &\mathbb{E}\|f(v^t)\| - \frac{\eta}{2}\mathbb{E}\|\nabla f(v^t)\|^2 + \frac{\eta}{2}\mathbb{E}\|f(v^t) - \frac{1}{N}\sum_{i=1}^N \nabla f_i(x^t_i)\|^2 + \frac{L\eta^2\sigma^2}{2N} \nonumber \\
    \leq &\mathbb{E}\|f(v^t)\| - \frac{\eta}{2}\mathbb{E}\|\nabla f(v^t)\|^2 + \frac{\eta}{2N}\sum_{i=1}^N\mathbb{E}\|f(v^t) - \nabla f_i(x^t_i)\|^2 + \frac{L\eta^2\sigma^2}{2N} \nonumber \\ 
    \leq &\mathbb{E}\|f(v^t)\| - \frac{\eta}{2}\mathbb{E}\|\nabla f(v^t)\|^2 + \frac{\eta L^2}{2N}\sum_{i=1}^N\mathbb{E}\|v^t - x^t_i\|^2 + \frac{L\eta^2\sigma^2}{2N} \nonumber \\
    \leq &\mathbb{E}\|f(v^t)\| - \frac{\eta}{2}\mathbb{E}\|\nabla f(v^t)\|^2 + \frac{9 \eta^3 G^2 L^2}{N}\sum_{k=1}^K R_k \Delta^2_k + \frac{L\eta^2\sigma^2}{2N}.
\end{align}
Dividing both sides by $\frac{\eta}{2}$ and rearranging the terms, we have
\begin{align}
    \mathbb{E}\|\nabla f(v^t)\|^2 \leq & \frac{2}{\eta}\left(\mathbb{E}\|f(v^t)\| - \mathbb{E}\|f(v^{t+1}\|\right)  + \frac{18 \eta^2 G^2 L^2}{N}\sum_{k=1}^K R_k \Delta^2_k + \frac{L\eta\sigma^2}{N}.
\end{align}
Taking the sum over $t = 0, 1, \cdots, T-1$ and dividing both sides by $T$, we have
\begin{align}
    \frac{1}{T}\sum_{t=0}^{T-1}\mathbb{E}\|\nabla f(v^t)\|^2 &\leq  \frac{2}{\eta T}\left(\mathbb{E}\|f(v^0)\| - \mathbb{E} \|f(v^T)\|\right) 
    + \frac{18 \eta^2 G^2 L^2}{N}\sum_{k=1}^K R_k \Delta^2_k + \frac{L\eta\sigma^2}{N} \nonumber \\
    &\leq \frac{2}{\eta T} (f(x^0) - f^*) + \frac{18 \eta^2 G^2 L^2}{N}\sum_{k=1}^K R_k \Delta^2_k + \frac{L\eta\sigma^2}{N},
    \label{virtual-bound}
\end{align}
where $f^* = \min_x f(x)$. Now, for the real sequence, we have
\begin{align}
    \frac{1}{T}\sum_{t=0}^{T-1}\mathbb{E}\|\nabla f(x^t)\|^2 
    \leq & \frac{1}{T}\sum_{t=0}^{T-1}(2\mathbb{E}\|\nabla f(v^t)\|^2 + 2\mathbb{E}\|\nabla f(v^t) - \nabla f(x^t)\|^2) \nonumber\\
    \leq & \frac{1}{T}\sum_{t=0}^{T-1}(2\mathbb{E}\|\nabla f(v^t)\|^2 + 2L^2\mathbb{E}\|v^t - x^t\|^2) \nonumber\\
    \leq &\frac{2}{T}\sum_{t=0}^{T-1}\mathbb{E}\|\nabla f(v^t)\|^2 + \frac{8\eta^2 G^2 L^2}{N} \sum_{k=1}^K R_k \Delta^2_k.
\end{align}
Substituting the bound on the virtual sequence into the above equation yields the claimed bound on the real sequence. 

\end{document}